\documentclass{article}
\usepackage[utf8]{inputenc}
\usepackage{amsmath,amssymb,amsfonts,amsthm}
\usepackage[english]{babel}
\usepackage{parskip}
\usepackage{hyperref}
\usepackage{tikz}
\usepackage{graphicx}
\usepackage{float}
\usepackage{subcaption}
\usepackage{fancyhdr} 
\usepackage{geometry}
\usepackage{setspace}
\usepackage{csquotes}
\usepackage{comment}
\usepackage{authblk}
\usepackage{framed}
\usepackage[onelanguage, ruled, linesnumbered]{algorithm2e}

\newtheorem{definition}{Definition}[section]
\newtheorem{lemma}{Lemma}[section]
\newtheorem{theorem}{Theorem}[section]

\newcommand{\keywords}[1]{%
\par\noindent\textbf{Keywords: }#1
}

\newcommand*{\N}{\mathbb{N}}

\newcommand*{\Z}{\mathbb{Z}}

\usepackage{hyperref}
\hypersetup{
  colorlinks=true,
  linkcolor=black,
  citecolor=black,
  filecolor=black,
  urlcolor=black,
}
\usepackage[shortlabels]{enumitem} 
\usepackage[capitalize]{cleveref} 

\usepackage{comment}

\usepackage{newpxtext}

\usepackage{microtype}

\usepackage[author=]{fixme}
\fxusetheme{color}

\usepackage{relsize}
\usepackage{tikz}
\usetikzlibrary {arrows.meta}
\newcounter{row}
\newcounter{col}
\newcounter{rows}
\newcounter{cols}

\newcommand\setrow[1]{
  \setcounter{col}{1}
  \foreach \n in {#1} {
    \edef\x{\value{col} - 0.5+\value{cols}}
    \edef\y{-\value{row} - 0.5 + \value{rows}}
    \node[anchor=center] at (\x, \y) {\n};
    \stepcounter{col}
  }
  \stepcounter{row} 
}

\title{Complexity of Fungal Automaton Prediction}
\author{E.~Formenti\thanks{Universit\'e C\^ote d'Azur, CNRS, I3S, Nice, France.}
,
E.~Goles\thanks{Facultad de Ingeniería y Ciencias, Universidad Adolfo Ibáñez, Santiago, Chile.}
,
K.~Perrot\thanks{Aix Marseille Universit\'e, CNRS, LIS, Marseille, France.}
,
M.~Ríos-Wilson\thanks{Facultad de Ingeniería y Ciencias, Universidad Adolfo Ibáñez, Santiago, Chile.}
~and
D.~Ruiz-Tala\thanks{Departamento de Ingeniería Matemática, Facultad de Ingeniería y
Ciencias, Universidad de Chile, Santiago, Chile. E-mail: domruiz123@gmail.com}}
\date{March 2026}

\begin{document}
\maketitle

\begin{abstract}
  Fungal automata are a nature-inspired computational model,
  where a rule is alternatively applied verticaly and horizontaly.
  In this work we study the computational complexity of predicting
  the dynamics of all fungal freezing totalistic one-dimentional rules of radius $1$,
  exhibiting various behaviors. Despite efficiently predictable in most cases
  (with non-deterministic logspace algorithms),
  a non-linear rule is left open to characterize.
  We further explore the freezing majority rule (which is totalistic), and prove that at radius $1.5$
  it becomes $\mathbf{P}$-complete to predict.
\end{abstract}

\section{Introduction}

Recently, Goles et al.~\cite{goles2020computational} proposed a variant of sandpile automata inspired by another class of automata known as \emph{fungal automata}~\cite{adamatzky2023fungal}. Fungal automata are a nature-inspired computational model that mimics the behavior of septa in fungi. Septa are compartments located within the cytoplasm of fungal organisms that regulate the flow between adjacent cells through structures called Woronin bodies. These structures react to external conditions and can either close—preventing flow between septa—or open to allow communication.

In this paper we consider a general formulation of fungal automata as a model defined on the two-dimensional grid, where cells are updated according to either their horizontal or vertical neighbors following a fixed update pattern.

A natural starting point for the systematic study of fungal automata is the class of \emph{totalistic rules}, that is, rules in which the localupdate function depends only on the sum of the states in the neighborhood rather than on their orientation. These rules are particularly convenient because the dependence on the sum of neighbor states can be extended naturally to higher dimensions, similarly to how the sandpile model admits higher-dimensional generalizations, even though the one-dimensional version of that model is not totalistic.

Among totalistic rules, we focus on \emph{freezing one-dimensional rules}, which have been studied in the context of automata networks. In this setting, a rule is called freezing if once a cell reaches state $1$, it never returns to state $0$.

In this work we define a prediction problem for fungal automata and study its computational complexity for all freezing totalistic one-dimensional rules of radius $1$. Among these rules, a particularly interesting case is the \emph{freezing majority rule}. For this rule we analyze the prediction problem for the fungal automaton associated with the radius $1.5$ freezing majority rule, and prove that its prediction problem is $\mathbf{P}$-complete.

\section{Circuit Value Problem}\label{sec:CVP}
\begin{definition}
    For every $n \in \mathbb{N}$, a \emph{Boolean circuit} $C$ with $n$ inputs is a Directed Acyclic Graph (DAG). It contains $n$ nodes with no incoming edges, called \emph{input nodes}, and a single node with no outgoing edges, called the \emph{output node}.
All other nodes are called \emph{gates} and are labeled with one of the logical operations $\vee$, $\wedge$, or $\neg$ (OR, AND, and NOT).
The $\vee$ and $\wedge$ gates have \emph{fan-in} (number of entering edges) $2$, while the $\neg$ gates have \emph{fan-in} $1$. A circuit \emph{monotone} if it only contains $\vee$ and $\wedge$ gates.
\end{definition}

With some abuse of notation, a boolean circuit as defined above implements a function $C:\{0,1\}^n \to\{0,1\}$ in a natural way stated in the next paragraph.

Since the circuit graph is acyclic, we can associate an integer \emph{circuit depth} to each node, such that all incoming edges of a node originate from nodes of strictly smaller circuit depth.
Given an input string $x \in {0,1}^n$, a Boolean value is assigned to each input node. The value of every other node is then uniquely determined by processing the nodes in increasing order of circuit depth. For each node, we examine the values assigned to its predecessor nodes and apply the Boolean operation labeling the node. In this way, each node receives a value in ${0,1}$. The value assigned to the output node is the output of the circuit on input $x$, which we denote by $C(x)$.

A classical decision problem related to Boolean circuits is the \emph{Circuit Value Problem} (CVP), formally defined as follows.

\begin{framed}\label{cvp}
\noindent\textbf{Circuit Value Problem (CVP)}\\
\textbf{Input:} A Boolean circuit $C$ with $n$ inputs and an assignment $x \in \{0,1\}^n$.\\
\textbf{Question:} Does $C(x) = 1$?
\end{framed}

It is well known that the Circuit Value Problem is \textbf{P-complete} \cite{greenlaw1995limits}. Moreover, the restriction of the problem to monotone circuits is known to be as hard as the general case. In particular, the following restricted version is \textbf{P}-complete \cite{goldschlager1977monotone}.

\begin{framed}
\noindent\textbf{Monotone Circuit Value Problem (MCVP)}\\
\textbf{Input:} A Monotone Boolean circuit $C$ with $n$ inputs and an assignment $x \in \{0,1\}^n$.\\
\textbf{Question:} Does $C(x) = 1$?
\end{framed}

The reduction of the prediction problem for automata to the CVP, relies on \emph{gadgets}. Gadgets are configurations of our CA that abstractly compute the output of logical gates. Working with monotone circuits constitutes a significant simplification of the construction, since it requires the emulation of fewer types of gadgets. Indeed, instances of MCVP are simpler, as they involve one fewer type of gate. In particular, MCVP will be used in the proof of \Cref{teo2}.

\subsubsection{Circuit Encoding.}
Some care must be taken in specifying the format in which a circuit $C$ is described. This observation is useful when proving that a reduction from the CVP can be implemented by a log-space transducer. Following the construction in \cite{ruzzo1981uniform}, each gate of a circuit is represented by a $4$-tuple
\[(g, t, g_1, g_2),\]
where $g$ is the number of the gate, $t$ specifies the type of the gate, and $g_1$ and $g_2$ are the numbers of the gates producing the inputs to gate $g$. All numbers are encoded in binary. If gate $g$ has only one input then, by convention, we set $g_2 = g_1$.

Without loss of generality, the input gates are numbered from $1$ to $n$. Since their predecessors are never used, we assume $g_1 = g_2 = 0$ for these gates. All remaining gates are numbered consecutively starting from $n+1$, and the inputs of each gate $g$ must come from gates with strictly smaller numbers. The gate numbered $n+i$ will occasionally be denoted by $G_i$.

Following Ruzzo \cite{ruzzo1981uniform}, the description $B$ of a complete circuit is defined as the concatenation of the descriptions of all its gates, sorted by increasing gate numbers.

\section{Cellular Automata}\label{sec:ca}

\begin{definition}
    A $d$-dimensional \emph{cellular automaton} is a tuple $(Q,N,f)$ where $Q\subseteq \Z$ is a finite set called the \emph{alphabet}, $N$ is finite subset of $\Z^d$ called \emph{neighborhood} and $f:Q^N\mapsto Q$ is the \emph{local rule} of the cellular automaton. A \emph{configuration} is an assigment of states to each cell $c: \Z^d \mapsto Q$ is called . If $S\subseteq \Z^d$, we denote by $c|_{S}$ to restriction of the configuration $c$ to the set $S.$ The local function $f$ of a cellular automaton induces a global function $F:Q^{\Z^d} \mapsto  Q^{\Z^d}$  given by $F(c)(x) = f(c |_{x+N}),\, \forall x \in \Z^d, c \in Q^{\Z^d}$. 
\end{definition}

In our study we will mainly deal with the \emph{von Neumman} neighborhood, i.e., $N = \{(x_1,...,x_d) \in \mathbb{Z}^d \mid |x_1| + ... + |x_d| \leq 1\}$. 

\begin{definition}
A cellular automaton $(Q,N,f)$ is \emph{totalistic} if there exists a
function $h:\mathbb{N}\to Q$ such that for every local configuration
$c'\in Q^N$ we have
\[
f(c')=h\!\left(\sum_{u\in N}c'(u)\right).
\]
\end{definition}
In other words, a CA is totalistic if its local rule depends only on the sum of the states in the neighborhood.

\begin{definition}
A cellular automaton $(Q,N,f)$ is said to be \emph{freezing} if there exists a partial order $\leq$ on $Q$ such that for every configuration $c\in Q^{\mathbb{Z}^d}$ and every cell $x\in\mathbb{Z}^d$,
\[
c(x)\leq F(c)(x).
\]
\end{definition}

In our case of interest, the set of states is restricted to $Q = \{0,1\}$. Therefore, a CA is \emph{freezing} if any cell that reaches state $1$ can never return to state $0$.

\begin{definition}
Let $(Q,N,f)$ be a freezing cellular automaton, and let $F:Q^{\mathbb{Z}^d}\to Q^{\mathbb{Z}^d}$ be its global function. A finite set $S\subseteq \mathbb{Z}^d$ together with a function $a:S\to Q$ is called an \emph{alliance} if for every configuration $c\in Q^{\mathbb{Z}^d}$ satisfying $c|_{S}=a$, for every $t\in\mathbb{N}$ and for every $x\in S$, we have
\[ F^{t}(c)(x)=a(x). \]
\end{definition}

In other words, alliances on a configuration are a set of cells $S$ in which the states remain unchanged along all iterations, independently of the states of cella outside $S$. We then say that a configuration $c$ \emph{has an alliance} on $S\subseteq \mathbb{Z}^d$ if the pair $(S,c|_{S})$ is an alliance for the cellular automaton $(Q,N,f)$.

We introduce two important notions regarding configurations.

\begin{definition}\label{def:connectedcomp}
Let $c \in Q^{\mathbb{Z}^2}$ be a configuration and let $x \in \mathbb{Z}^2$ be a cell such that $c(x)=0$. The \emph{connected component} of $x$ is defined as the set of all cells $y \in \mathbb{Z}^2$ for which there exist $l \in \mathbb{N}$ and cells $x_1,\dots,x_l \in \mathbb{Z}^2$ such that
\[ x_1 = x, \quad x_l = y, \] and for every $i \in \{1,\dots,l\}$,
\[ c(x_i)=0 \quad \text{and} \quad x_{i-1} \text{ is adjacent to } x_i. \]
\end{definition}

An example of a connected componnent is on \Cref{fig:compocon}. In what follows, we define the notion of \emph{staircase}, a set of cells that arises naturally in the analysis of radius-$1$ fungal automata.

\begin{definition}
A \emph{staircase} is a sequence of cells $x_0...x_n$ such that there exists $\lambda_1,\lambda_2$ such that either (1) or (2) occurs: \begin{align*}
    && \text{(1) } & \forall i \in [n], x_i = x_0 + \lambda_1\lceil\frac{i}{2}\rceil (1,0) + \lambda_2\lfloor\frac{i}{2}\rfloor (0,1)\\
    && \text{(2) } & \forall i \in [n], x_i = x_0 + \lambda_1\lfloor\frac{i}{2}\rfloor (1,0) + \lambda_2\lceil\frac{i}{2}\rceil (0,1)
\end{align*} All cells in the staircase have two adjacent cells belonging to the staircase,  except for two cells, which we call \emph{endpoints}. Example of a staircase is in \Cref{fig:staircase}.
\end{definition}

\begin{figure}[t]
\centering\resizebox{0.2\textwidth}{!}{
\begin{tikzpicture}[
  every node/.append style={font=\relsize{+3}}
]

\draw[draw opacity=0,
      fill={rgb,255:red,184; green,233; blue,134},
      fill opacity=1]
      (1,3) rectangle (4,4);

\draw[draw opacity=0,
      fill={rgb,255:red,184; green,233; blue,134},
      fill opacity=1]
      (1,1) rectangle (3,3);

    \draw[line width=0.1mm] (0, 0) grid (5, 5);
    \setcounter{cols}{0}
    \setcounter{rows}{5}
    \setcounter{row}{0}
    \setrow {1,1,1,1,1}
    \setrow {1,0,0,0,1}
    \setrow {1,0,$x$,1,1}
    \setrow {1,0,0,1,1}
    \setrow {1,1,1,1,1};

\end{tikzpicture}}
\caption{In green, the connected component of cell $x$.}
\label{fig:compocon}
\end{figure}

\begin{figure}[t]
\centering
\resizebox{0.2\textwidth}{!}{\begin{tikzpicture}[
  every node/.append style={font=\relsize{+3}}]


    
    \draw[draw opacity=0, fill={rgb,255:red,184; green,233; blue,134}, fill opacity=0.43] (3,2) -- (5,2) -- (5,3) -- (3,3) -- cycle;
    \draw[draw opacity=0, fill={rgb,255:red,184; green,233; blue,134}, fill opacity=0.43] (4,3) -- (6,3) -- (6,4) -- (4,4) -- cycle;
    \draw[draw opacity=0, fill={rgb,255:red,184; green,233; blue,134}, fill opacity=0.43] (5,4) -- (7,4) -- (7,5) -- (5,5) -- cycle;

    \draw[draw opacity=0, fill={rgb,255:red,184; green,233; blue,134}, fill opacity=0.43] (6,5) -- (8,5) -- (8,6) -- (6,6) -- cycle;

    \draw[draw opacity=0, fill={rgb,255:red,184; green,233; blue,233}, fill opacity=0.83] (7,6) -- (8,6) -- (8,7) -- (7,7) -- cycle;
    \draw[line width=0.1mm] (2, 1) grid (9, 8);

\end{tikzpicture}}
\caption{The set of cells forming a staircase is shown in green. The size of the staircase is not fixed. All cells in the staircase have two adjacent cells in the staircase, except for the endpoints. The endpoints do not have a fixed orientation; therefore, the extra cell shown in blue may or may not belong to the staircase.}
\label{fig:staircase}
\end{figure}
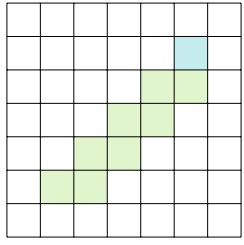

\section{Automata Networks}

Let $[n] = \{1,...,n\}$. We define the concepts necessary to discuss one of the prediction problems that remained open in \Cref{sec:f2}.

\subsubsection{Tree Width}
\begin{definition}
Given a graph $G=(V,E)$, a \emph{tree decomposition} is a pair
$\mathcal{D}=(T,\Lambda)$ such that $T$ is a tree graph and $\Lambda$ is a family
of subsets of nodes $\Lambda=\{X_t \subseteq V \mid t\in V(T)\}$, called
\emph{bags}, such that:
\begin{itemize}
    \item Every node in $G$ is in some $X_t$, i.e.,
    \[
    \bigcup_{t\in V(T)} X_t = V.
    \]
    \item For every $e=uv\in E$ there exists $t\in V(T)$ such that
    $u,v\in X_t$.
    \item For every $u\in V$ and every $t,v\in V(T)$, if $w\in V(T)$ is on the
    $t$--$v$ paht in $T$, then
    \[
    X_t \cap X_v \subseteq X_w.
    \]
\end{itemize} 

We define the width of a tree decomposition $\mathcal{D}$ as
\[ \mathrm{width}(\mathcal{D})=\max_{t\in V(T)} |X_t| - 1. \]
\end{definition}

\begin{definition}
Given a graph $G=(V,E)$, we define its \emph{treewidth} as the parameter \[ \mathrm{tw}(G)=\min_{\mathcal{D}} \mathrm{width}(\mathcal{D}). \] In other words, the treewidth is the minimum width of a tree decomposition of $G$. Note that if $G$ is a connected graph such that $|E(G)| \ge 2$, then $G$ is a tree if and only if $\mathrm{tw}(G)=1$.
\end{definition}

\subsubsection{Automata networks}
Automata networks (AN) are a generalization of cellular automata in which the underlying grid is replaced by a graph, and the update function depends on the neighbors of each node in the graph. Formally:

\begin{definition}
We define a automata network over the alphabet $Q\subseteq\Z$ and a graph $G = (V = [n],E \subseteq\binom{[n]}{2})$ as a tuple $(G,\mathcal{F}=\{F_v:Q^{|N(v)|}\to Q\,|\,v\in V\})$. As for CA, we call configurations elements of $Q^n$ and we define the global transition $F:Q^n\to Q^n$ defined by
\[ F(x)=y\in Q^n \text{ with } y_v = F_v(x),\ \forall v. \]
\end{definition}

Analogous to the definitions of freezing CA we can define the notion of freezing AN.
\begin{definition}
We say that a AN $(G,\mathcal{F})$ with global function $F$ defined over the alphabet $Q$ is \emph{freezing}, if there exists a partial order $\leq$ in $Q$ such that for every configuration $c\in Q^n$ and $t\in\mathbb{N}$ we have
\[ c(i) \le F(c)(i) \quad\text{for every } i \in [n] \]
\end{definition}

Following the ideas from \cite{paramcom}, \emph{AN prediction} problem is formally defined as follows:

\begin{framed}\label{ANpred}
\noindent\textbf{AN prediction}\\
\textbf{Input:} An AN $(G,\mathcal{F})$, the initial configuration $c$ and a node $v  \in V$ (encoded in binary).\\
\textbf{Question:} Suppose the initial configuration of the FA $c$. Will cell $x$ assume a non-zero state at some configuration $F_Z(c,t) \text{ with }t \in \N$?
\end{framed}

In \cite{paramcom}, it is shown that:
\begin{theorem}\label{teof2}
The \emph{AN prediction} problem lies in \textbf{NC} for freezing AN over graphs with bounded treewidth and degree.
\end{theorem}

\section{Fungal Automata}\label{secfunggen}
Let $\mathcal{A}=(Q,N,f)$ be a one-dimensional (1d) cellular automata with state set $Q\subseteq\mathbb{N}$ and neighborhood  $N = \{n_1,\dots,n_r\}\subset \mathbb{Z}$. Let $\Phi:\mathbb{N}\to\{H,V\}$ be a periodic function with period $k$, which we identify with a word $Z = Z_0...Z_{k-1}$ over the alphabet $\{H,V\}$ such that $\Phi(i) = Z_{i \bmod k},\forall i \in \mathbb{N}$.
Define the two-dimensional neighborhoods:
\[
\mathcal{N}_H := \{(n_i,0)\mid i\in[k]\}
\quad\text{and}\quad
\mathcal{N}_V := \{(0,n_i)\mid i\in[k]\}.
\]

Let $f_H:Q^{\mathcal{N}_H}\to Q$ be the local rule that applies $f$ according to the neighborhood $\mathcal{N}_H$, and let $f_V:Q^{\mathcal{N}_V}\to Q$ be the local rule that applies $f$ according to the neighborhood $\mathcal{N}_V$. 
That is, for any $c'\in Q^{\mathbb{Z}^2}$:
\begin{align*}
f_H(g) &= f\big(c'(n_1,0),\dots,c'(n_r,0)\big),\\
f_V(g) &= f\big(c'(0,n_1),\dots,c'(0,n_r)\big).
\end{align*}

The \emph{fungal automaton} (FA) associated with $\mathcal{A}$ is the tuple
$
(f,Q,M,Z,f_H,f_V),
$
where $M=\mathcal{N}_h\cup\mathcal{N}_v$.
Each local function $f_D$ with $D\in \{H,V\}$ induces a global function 
$F_D:Q^{\mathbb{Z}^2}\to Q^{\mathbb{Z}^2}$. Like a classical CA, in their turn these global functions allow to define the global transition function 
$F_Z:Q^{\mathbb{Z}^2}\to Q^{\mathbb{Z}^2}$. That is recursively defined as
$F^0(c) = c$, and 
$F^{t+1}(c) = F_{\Phi(t+1)}\big(F^t(c)\big)$
for all $c\in Q^{\mathbb{Z}^2}$ and $t\in\mathbb{N}$. As for classical CA, elements $c\in Q^{\mathbb{Z}^2}$ are called \emph{configurations}.


Since FA are dynamical systems, it makes sense to define a prediction problem that we call the \emph{FA prediction problem} (FA-pred), formally defined as follows:

\begin{framed}\label{FApred}
\noindent\textbf{FA prediction problem (FA-pred)}\\
\textbf{Parameters:} A 1d CA $(f,Q,M)$, a word $Z \in \{H,V\}^+$, a value $s\in Q$.\\
\textbf{Input:} A rectangle $R\subseteq \mathbb{Z}^2$, a configuration $c|_R$ and the index $x$ of a cell on $R$ (encoded in binary).\\
\textbf{Question:} Will $x$ assume a state different from $c(x)$ when the FA $(f,Q,M,Z,f_H,f_V)$ is initialized with $c|_R$ surrounded by cells in state $s$?
\end{framed}

We will refer to $s$ as the \emph{background state}, and to $x$ as the \emph{target cell}.

\section{Fungal Automata from totalistic freezing 1d rules} 

In this section, we study FA by focusing on those defined by one-dimensional \emph{totalistic freezing} rules of radius 1.
For each one-dimensional local rule $f$ with state set $Q$, and for a word $Z \in \{H,V\}^+$ with $Z = Z_0...Z_k$, we recall from \Cref{secfunggen} that the associated fungal automata (FA) is the tuple $(f,Q,M,Z,f_H,f_V)$, where the global function $F_Z$ defines the dynamics of the FA alternating $f_H$ and $f_V$. Recall from the definition that $f_H$ and $f_V$ update cells according to the rule $f$, considering horizontal and vertical neighbors (See \Cref{fig:ejemplofungal}), respectively. In this section, we restrict our attention to fungal automata for which $Z = HV$. We will denote the corresponding global rule by $F_{HV}$ in the sequel, without explicitly recalling the underlying one-dimensional rule, in order to avoid an excessive use of notation.

\begin{figure}[t]
\centering
\resizebox{\textwidth}{!}{\begin{tikzpicture}[
  every node/.append style={font=\relsize{+3}}
]

\draw[draw opacity=0,
      fill={rgb,255:red,184; green,233; blue,134},
      fill opacity=0.43]
      (1,2) rectangle (4,3);

\draw[draw opacity=0,
      fill={rgb,255:red,184; green,233; blue,134},
      fill opacity=1]
      (2,2) rectangle (3,3);

\draw[draw opacity=0,
      fill={rgb,255:red,184; green,233; blue,134},
      fill opacity=0.43]
      (9,2) rectangle (12,3);

\draw[draw opacity=0,
      fill={rgb,255:red,184; green,233; blue,134},
      fill opacity=1]
      (10,2) rectangle (11,3);
\draw[draw opacity=0,
      fill={rgb,255:red,184; green,233; blue,134},
      fill opacity=0.43]
      (18,1) rectangle (19,4);

\draw[draw opacity=0,
      fill={rgb,255:red,184; green,233; blue,134},
      fill opacity=1]
      (18,2) rectangle (19,3);

\draw[draw opacity=0,
      fill={rgb,255:red,184; green,233; blue,134},
      fill opacity=0.43]
      (26,1) rectangle (27,4);

\draw[draw opacity=0,
      fill={rgb,255:red,184; green,233; blue,134},
      fill opacity=1]
      (26,2) rectangle (27,3);

\draw[-{Stealth[length=2mm]}, line width=0.4mm]
      (5.5,2.5) -- (7.5,2.5);

\draw[-{Stealth[length=2mm]}, line width=0.4mm]
      (21.5,2.5) -- (23.5,2.5);

    \draw[line width=0.1mm] (0, 0) grid (5, 5);

    \setcounter{rows}{5}
    \setcounter{row}{0}
    \setrow { , , , , }
    \setrow { , ,1, , }
    \setrow { ,0,0,1, }
    \setrow { , ,1, , }
    \setrow { , , , , };
    
    \draw[line width=0.1mm] (8, 0) grid (13, 5);

    \setcounter{cols}{8}
    \setcounter{rows}{5}
    \setcounter{row}{0}
    \setcounter{rows}{5}
    \setcounter{row}{0}
    \setrow { , , , , }
    \setrow { , ,1, , }
    \setrow { ,0,0,1, }
    \setrow { , ,1, , }
    \setrow { , , , , };
    
    \draw[line width=0.1mm] (16, 0) grid (21, 5);

    \setcounter{cols}{16}
    \setcounter{rows}{5}
    \setcounter{row}{0}
    \setcounter{rows}{5}
    \setcounter{row}{0}
    \setrow { , , , , }
    \setrow { , ,1, , }
    \setrow { ,0,0,1, }
    \setrow { , ,1, , }
    \setrow { , , , , };

    \draw[line width=0.1mm] (24, 0) grid (29, 5);

    \setcounter{cols}{24}
    \setcounter{rows}{5}
    \setcounter{row}{0}
    \setcounter{rows}{5}
    \setcounter{row}{0}
    \setrow { , , , , }
    \setrow { , ,1, , }
    \setrow { ,0,1,1, }
    \setrow { , ,1, , }
    \setrow { , , , , };

\setcounter{cols}{0}

\node at (2.5,-1) {$c$};
\node at (10.5,-1) {$F^1_{HV}(c)$};
\node at (18.5,-1) {$F^1_{HV}(c)$};
\node at (26.5,-1) {$F^2_{HV}(c)$};

\end{tikzpicture}}
\caption{Let $c$ be the configuration shown in the figure, where unlabeled cells are assumed to be in state $0$. Consider the one-dimensional cellular automaton $(f,\{-1,0,1\},\{0,1\})$, where $f(1,0,1)=1, \text{ and } f(a,b,c)=0 \text{ for all other values $a,b,c \in \{0,1\}$}$, and its associated fungal automaton $(f,Q,M,HV,f_H,f_V)$. The remaining grids depict successive updates of $c$ under the global function $F_{HV}$. Note that $F^1_{HV}(c)=F_H(c)$ and $F^2_{HV}(c)=F_V\left(F^1_{HV}(c)\right)$. All cells are updated simultaneously, but the central cell is highlighted in order to illustrate the behavior of the fungal automaton. No cell transitions to state $1$ during the update $F_H$. During the update $F_V$, only the central cell changes its state. We therefore focus on the evolution of the central cell. More precisely, in the leftmost grid the cell highlighted in dark green is updated according to the local rule $f$ using its horizontal neighborhood (cells highlighted in light green). Since one of its neighbors has value $0$, the cell remains in state $0$. In the rightmost grid, the same cell is updated according to $f$ with respect to its vertical neighbors, again highlighted in light green. Two identical grids are shown in order to emphasize the neighborhood on which the local rule $f$ is applied during each update to the central cell.}
\label{fig:ejemplofungal}
\end{figure}
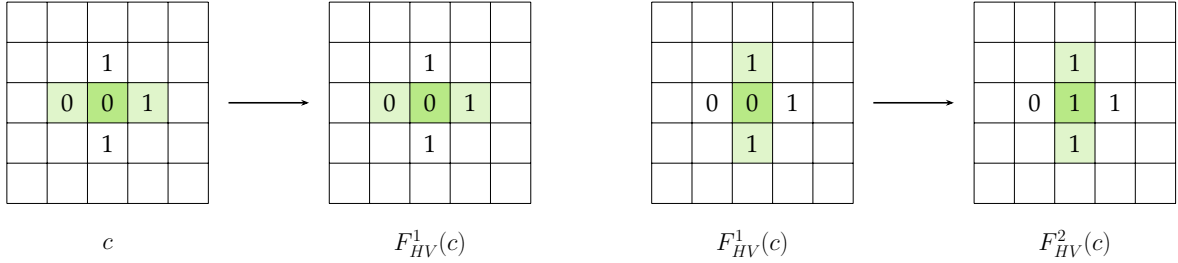

\subsection{Fungal automata assosiated with Totalistic Freezing CA of radius 1}

Fungal automata based on a totalistic freezing CA of radius 1, are such that their local rule $f$ is such that there exists a function $h:\{0,...,3\} \to \{0,1\}$ satisfying:

\begin{equation*}\label{eq:fteca}
f(x_1,x_2,x_3) =
\begin{cases}
    1 & \text{ if } x_2 = 1, \\
    h(x_1+x_3) & \text{ otherwise}.
\end{cases}
\end{equation*}

Therefore, the function $f$ is fully determined by the values of $h:\{0,...,3\}\to \{0,1\}$. Hence, a rule $f$ can be characterized by a triple $(h(0),h(1),h(2))$ of the function $h$ that defines $f$. We conclude that there are only 8 rules to study. They are shown in \Cref{tab:rules}. We study all the FA $\mathcal{F}_i = (f_i,\{0,1\},M,HV,F_{HV}), i \in \{0,...,8\}$ associated to each $f_i$, where $M$ denotes the von Neumann neighborhood. In the third column of \Cref{tab:rules}, we provide the complexity of the problem FA-pred for each $\mathcal{F}_i, \, i \in \{0,...,8\}$ for any background state $s \in \{0,1\}$. The corresponding proofs are given in the following sections.

\begin{table}[t]
\caption{Fungal automata studied in this section. The first two columns specify the one-dimensional rules under consideration, according to the triple defining the totalistic function $h$. The third column provides the complexity of the prediction problem.}
\centering
\begin{tabular}{|c|c|c|}
\hline
Rule & $h$ vector & FA-pred Complexity \\ \hline
$f_0$ & $(0,0,0)$ &  $\mathbf{AC}^0$ (\Cref{lemf07})\\ \hline
$f_1$ & $(0,0,1)$ &  \textbf{NL} (\Cref{lem:f1})\\ \hline
$f_2$ & $(0,1,0)$ &  unknown \\ \hline
$f_3$ & $(0,1,1)$ &  $\mathbf{AC}^0$ (\Cref{lemf3})\\ \hline
$f_4$ & $(1,0,0)$ &  $\mathbf{AC}^0$ (\Cref{lemf4})\\ \hline
$f_5$ & $(1,0,1)$ &  $\mathbf{AC}^0$ (\Cref{lemf5})\\ \hline
$f_6$ & $(1,1,0)$ &  $\mathbf{AC}^0$ (\Cref{lemf6})\\ \hline
$f_7$ & $(1,1,1)$ &  $\mathbf{AC}^0$ (\Cref{lem:f1})\\ \hline
\end{tabular}

\label{tab:rules}
\end{table}


\subsection{Computational Complexity}
We begin by studying the FA-pred from rules $f_i,\, i \in \{0,...,8\}$ and $Z = HV$. We carry out this analysis independently of the chosen background state, since the properties we expose characterize the dynamics regardless of this value.

\begin{lemma}\label{lemf07}
  For the rules $f_0$ and $f_7$ FA-pred $\in \mathbf{AC}^0$.
\end{lemma}

\begin{proof}
  These rules are trivial. Rule $f_0$ corresponds to the identity rule, so the
  configuration never changes. Rule $f_7$ maps every cell in state $0$ to $1$ in
  a single update. In both cases, the state of the target cell can be determined
  by seeing the initial state of the objective cell.
\end{proof}

\begin{lemma}\label{lemf3}
  For the rule $f_3$ FA-pred $\in \mathbf{AC}^0$.
\end{lemma}

\begin{proof}
  This rule corresponds to the logical \emph{or} operation in one-dimension. Every cell in state $0$ in a configuration $c$ that has a horizontal neighbor in state $1$ takes value $1$ under $f_H(c)$; analogously, every cell with a vertical neighbor in state $1$ takes value $1$ under $f_V(c)$ (See \Cref{fig:ejemplofungalf3}). Hence, if a configuration contains at least one cell in state $1$, then every cell eventually reaches state $1$. Therefore, the prediction problem reduces to deciding whether there exists at least one cell in the initial rectangle with value $1$, which can be done with a unique OR having as inputs all the bits from the initial configuration.
\end{proof}

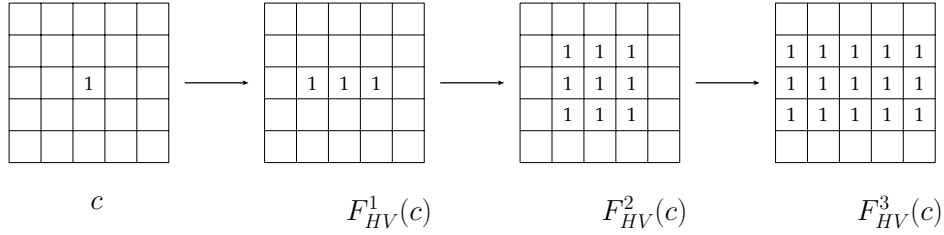
\begin{figure}[t]
\centering
\resizebox{0.8\textwidth}{!}{\begin{tikzpicture}[
  every node/.append style={font=\relsize{+3}}
]

\draw[-{Stealth[length=2mm]}, line width=0.4mm]
      (5.5,2.5) -- (7.5,2.5);
\draw[-{Stealth[length=2mm]}, line width=0.4mm]
      (13.5,2.5) -- (15.5,2.5);
      
\draw[-{Stealth[length=2mm]}, line width=0.4mm]
      (21.5,2.5) -- (23.5,2.5);

    \draw[line width=0.1mm] (0, 0) grid (5, 5);

    \setcounter{rows}{5}
    \setcounter{row}{0}
    \setrow { , , , , }
    \setrow { , , , , }
    \setrow { , ,1, , }
    \setrow { , , , , }
    \setrow { , , , , };
    
    \draw[line width=0.1mm] (8, 0) grid (13, 5);

    \setcounter{cols}{8}
    \setcounter{rows}{5}
    \setcounter{row}{0}
    \setcounter{rows}{5}
    \setcounter{row}{0}
    \setrow { , , , , }
    \setrow { , , , , }
    \setrow { ,1,1,1, }
    \setrow { , , , , }
    \setrow { , , , , };
    
    \draw[line width=0.1mm] (16, 0) grid (21, 5);

    \setcounter{cols}{16}
    \setcounter{rows}{5}
    \setcounter{row}{0}
    \setcounter{rows}{5}
    \setcounter{row}{0}
    \setrow { , , , , }
    \setrow { ,1,1,1, }
    \setrow { ,1,1,1, }
    \setrow { ,1,1,1, }
    \setrow { , , , , };

    \draw[line width=0.1mm] (24, 0) grid (29, 5);

    \setcounter{cols}{24}
    \setcounter{rows}{5}
    \setcounter{row}{0}
    \setcounter{rows}{5}
    \setcounter{row}{0}
    \setrow { , , , , }
    \setrow {1,1,1,1,1}
    \setrow {1,1,1,1,1}
    \setrow {1,1,1,1,1}
    \setrow { , , , , };

\setcounter{cols}{0}

\node at (2.5,-1) [anchor=north west,scale=1.5][inner sep=0.75pt] {$c$};
\node at (10.5,-1) [anchor=north west,scale=1.5][inner sep=0.75pt] {$F^1_{HV}(c)$};
\node at (18.5,-1) [anchor=north west,scale=1.5][inner sep=0.75pt] {$F^2_{HV}(c)$};
\node at (26.5,-1) [anchor=north west,scale=1.5][inner sep=0.75pt] {$F^3_{HV}(c)$};

\end{tikzpicture}}
\caption{Successive updates of the fungal automaton produced by the rule $f_3$.}
\label{fig:ejemplofungalf3}
\end{figure}

\begin{lemma}\label{lemf4}
  For the rule $f_4$ FA-pred $\mathbf{AC}^0$.
\end{lemma}

\begin{proof}
  For this rule, in any configuration $c$ of the 2-dimensional grid, only cells whose two horizontal neighbors are both in state $0$ may switch from state $0$ to $1$ under the application of $f_H$. Consequently, in the configuration $F_H(c)$ there are no cells in state $0$ whose two horizontal neighbors are also in state $0$, since any such cell would have transitioned to state $1$. As no cell can change from $1$ back to $0$, no new cells with this property can appear.

  By an analogous argument, in the configuration $F^2_{HV}(c) = F_V(F_H(c))$ there are no cells in state $0$ whose two vertical neighbors are in state $0$. Hence, after two iterations the dynamics stabilizes, that is, $F^t_{HV}(c)=F^2_{HV}(c)$ for all $t\geq 2$.
  
  Therefore, the prediction problem only requires inspecting the behavior of the target cell $x$ during the first two iterations of $F_{HV}$, and then deciding whether its state differs from the initial one.
 
  Given the initial configuration $c$, we have $F^1_{HV}(c) = F_H(c)$. Thus, during the first iteration, the cell $x$ and its vertical neighbors are updated according to the local rule $f_4$, using their respective horizontal neighborhoods. Next, $F^2_{HV}(c) = F_V(F_H(c))$, so during the second iteration, the state of $x$ is updated according to the values of its vertical neighbors in the configuration $F_H(c)$. Consequently, predicting the behavior of $x$ after two iterations depends only on the Moore neighborhood of $x$ in the initial configuration. 
\end{proof}

\begin{lemma}\label{lemf5}
  For the rule $f_5$ FA-pred $\in \mathbf{AC}^0$.
\end{lemma}

\begin{proof}
  As in the case of $f_4$, for every configuration $c$, at
  $F^2_{HV}(c)$ no cell has both its horizontal neighbors or both its vertical
  neighbors in state $0$. At this point, the cells in state $0$ consist only of
  staircases and $2\times 2$ blocks (See \Cref{fig:ejemplofungalf5}).

  For this rule, all cells belonging to staircases eventually transition to state $1$ (See \Cref{fig:ejemplofungalf5}), while the cells forming $2\times 2$ blocks are stable. Hence, the prediction problem reduces to determining whether, after two iterations, the target cell $x$ belongs to a staircase or to a $2\times 2$ block.

  To decide this, it suffices to know the state after two iterations of all cells in the Moore neighborhood of $x$. After two iterations, the cell $x$ remains in state $0$ if and only if it belongs to a $2\times 2$ block; otherwise, it belongs to a staircase of $0$ and it eventually reaches state $1$.
\end{proof}

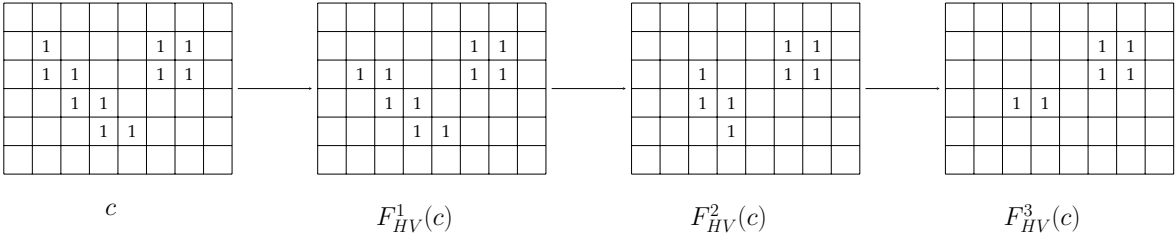
\begin{figure}[t]
\centering
\resizebox{\textwidth}{!}{\begin{tikzpicture}[
  every node/.append style={font=\relsize{+3}}
]

    \draw[line width=0.1mm] (0, 9) grid (8, 15);

    \setcounter{rows}{15}
    \setcounter{row}{0}
    \setrow { , , , , }
    \setrow { ,1, , , , 1,1,}
    \setrow { ,1,1, , , 1,1,}
    \setrow { , ,1,1, }
    \setrow { , , ,1,1}
    \setrow { , , , , }
;
    
    \draw[- stealth] (8.2,12) -- (10.8,12);

    \draw[line width=0.1mm] (11, 9) grid (19, 15);

    \setcounter{cols}{11}
    \setcounter{rows}{15}
    \setcounter{row}{0}
    \setrow { , , , , }
    \setrow { , , , , ,1 ,1,}
    \setrow { ,1,1, , ,1 ,1,}
    \setrow { , ,1,1, }
    \setrow { , , ,1,1}
    \setrow { , , , , }
;

    \draw[- stealth] (19.2,12) -- (21.8,12);
    \draw[line width=0.1mm] (22, 9) grid (30, 15);

    \setcounter{cols}{22}
    \setcounter{rows}{15}
    \setcounter{row}{0}
    \setrow { , , , , }
    \setrow { , , , , , 1,1,}
    \setrow { , ,1, , , 1,1,}
    \setrow { , ,1,1, }
    \setrow { , , ,1, }
    \setrow { , , , , }
;

    \draw[- stealth] (30.2,12) -- (32.8,12);
    \draw[line width=0.1mm] (33, 9) grid (41, 15);

    \setcounter{cols}{33}
    \setcounter{rows}{15}
    \setcounter{row}{0}
    \setrow { , , , , }
    \setrow { , , , , , 1,1,}
    \setrow { , , , , , 1,1,}
    \setrow { , ,1,1, }
    \setrow { , , , , }
    \setrow { , , , , }

\setcounter{cols}{0}

    \draw (3.5,8) node [anchor=north west,scale=1.5][inner sep=0.75pt]    {$c$};
    \draw (35,8) node [anchor=north west,scale=1.5][inner sep=0.75pt]    {$F^3_{HV}(c)$};
    \draw (13,8) node [anchor=north west,scale=1.5][inner sep=0.75pt]    {$F^1_{HV}(c)$};
    \draw (24,8) node [anchor=north west,scale=1.5][inner sep=0.75pt]    {$F^2_{HV}(c)$};
    
\end{tikzpicture}}
\caption{Configuration containing a staircase and a $2 \times 2$ block, both in state $0$, and its evolution over the FA of $f_5$. Unlabeled cells are assumed to be in state $1$. From this configuration, under the application of $f_H$, every cell having two horizontal neighbors are in state $1$ switches to state $0$. Analogously, under the application of $f_V$, every cell having two vertical neighbors are in state $1$ switches to state $0$.}
\label{fig:ejemplofungalf5}
\end{figure}

\begin{lemma}\label{lemf6}
  For the rule $f_6$ FA-pred $\in \mathbf{AC}^0$.
\end{lemma}

\begin{proof}
  For this rule, following an argument similar to that for $f_4$, we observe that for any configuration $c$, if a cell has one or two horizontal neighbors in state $0$, then it takes value $0$ in the configuration $F_H(c)$. Consequently, in $F_H(c)$ there are no cells with this property, since otherwise such a cell would have had fewer neighbors in state $0$ in the configuration $c$, contradicting the \emph{freezing} property of the dynamics. An analogous argument applies to $f_V(c)$. Therefore, as for $f_4$, the dynamics only requires analyzing two iterations.
\end{proof}

\subsection{Rule \texorpdfstring{$f_1$}{TEXT}}
\label{ss:f1}

In this section, we present a nondeterministic algorithm that solves the prediction problem for the FA associated with the rule $f_1$. To this end, we characterize the \emph{alliances} of this dynamics (see \Cref{sec:ca}). We note that for the rule $f_1$, given a configuration $c$, only cells whose two horizontal neighbors are both in state $1$ may change under the application of $f_H(c)$, and analogously, only cells whose two vertical neighbors are both in state $1$ may change under the application of $f_V(c)$. With this observation in mind, we identify alliances in $c$ with the following lemma.

\begin{algorithm}[t]
	\SetAlgoLined
	\KwIn{Initial rectangular finite support of the configuration $c$. Cell $x$}
	\KwOut{\texttt{accept} if $x$ belong to an alliance as defined in \Cref{characaliance}}
    
	\If{$x$ is on a $2\times2$ block of value $0$ in $c$}{
	    \texttt{accept}\;
	}
	
    \uElseIf{Exists a cell $a = (a_1,a_2)$ in a $2\times2$ block of value $0$}{
        Choose $\lambda_1,\lambda_2\in \{-1,1\}$\;
        $f \leftarrow \texttt{False}$\;
        $s_1,s_2 \leftarrow a_1,a_2$\;
        \While{$(s_1,s_2)$ is a coordinate inside the rectangle}{
            \textbf{if} $c(s_1,s_2)=1$ \textbf{then} \texttt{reject}\;
            \textbf{if} $(s_1,s_2)=x$ \textbf{then} $f \longleftarrow \texttt{True}$\;
            \textbf{if} $(s_1,s_2)$ is at $2\times2$ block of value $0$ and $f = \texttt{True} \textbf{ then } \texttt{accept} $\;
            \textbf{if} $c(s_1+\lambda_1,s_2) = 1$ \textbf{then} \texttt{reject}\;
            $(s_1,s_2) \leftarrow (s_1+\lambda_1,s_2+\lambda_2)$
            }}
    \Else{\texttt{reject}}
	\caption{Non-deterministic algorithm to decide if $x$ belong to an alliance as defined in \Cref{characaliance}}
	\label{alg:f1}
\end{algorithm}

\begin{lemma}\label{characaliance}
Inductively (See fig \Cref{fig:alianza} for an example):
\begin{enumerate}[itemsep=0pt, topsep=0pt, parsep=0pt]
    \item A $2 \times 2$ block with cells in state $0$ is an alliance and a staircase with an endpoint contained in a $2 \times 2$ block in state $0$ is an alliance.
    \item Any set of cells in state $0$ that is the union of sets forming alliances
    is itself an alliance.
\end{enumerate}
\end{lemma}

\begin{proof}
    Clearly, the first two types of cell sets, if initially in state $0$, are alliances. The fact that the union of alliances is again an alliance follows from the previous observation that, in this dynamics, a cell switches from $0$ to $1$ only if both of its neighbors are in state $1$. Therefore, enlarging the set of cells in state $0$ cannot cause a cell that initially did not have two neighbors in state $1$ to suddenly acquire two such neighbors. Cells that do not belong to any of the structures described above fall into one of two categories. Either they have, in some direction, no neighbors in state $0$, which implies that they will switch to state $1$ in the next iteration and are therefore not stable; or they are part of a staircase whose endpoints are not contained in $2\times 2$ blocks, a structure which is also not stable (Staircases without endpoints in a $2\times 2$ block behave the same way as for rule $f_5$ in \Cref{fig:ejemplofungalf5}).
\end{proof}

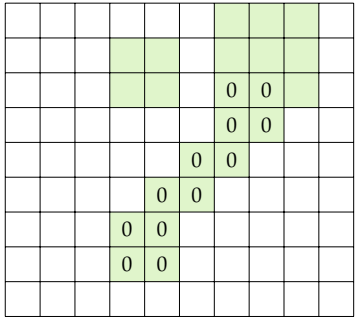
\begin{figure}[t]
\centering
\resizebox{0.3\textwidth}{!}{\begin{tikzpicture}[
  every node/.append style={font=\relsize{+3}}]

    \draw  [draw opacity=0][fill={rgb, 255:red, 184; green, 233; blue, 134 }  ,fill opacity=0.43 ] (6,6) -- (9,6) -- (9,9) -- (6,9) -- cycle ;

    \draw  [draw opacity=0][fill={rgb, 255:red, 184; green, 233; blue, 134 }  ,fill opacity=0.43 ] (3,6) -- (5,6) -- (5,8) -- (3,8) -- cycle ;
    
    \draw[draw opacity=0, fill={rgb,255:red,184; green,233; blue,134}, fill opacity=0.43] (3,1) -- (5,1) -- (5,3) -- (3,3) -- cycle;
    \draw[draw opacity=0, fill={rgb,255:red,184; green,233; blue,134}, fill opacity=0.43] (4,3) -- (6,3) -- (6,4) -- (4,4) -- cycle;
    \draw[draw opacity=0, fill={rgb,255:red,184; green,233; blue,134}, fill opacity=0.43] (5,4) -- (7,4) -- (7,5) -- (5,5) -- cycle;

    \draw[draw opacity=0, fill={rgb,255:red,184; green,233; blue,134}, fill opacity=0.43] (6,5) -- (8,5) -- (8,6) -- (6,6) -- cycle;
    \draw[line width=0.1mm] (0, 0) grid (10, 9);

    \setcounter{cols}{0}
    \setcounter{rows}{9}
    \setcounter{row}{0}
    \setrow { , , , , , , , , , }
    \setrow { , , , , , , , , , }
    \setrow { , , , , , ,0,0, , }
    \setrow { , , , , , ,0,0, }
    \setrow { , , , , ,0,0, , }
    \setrow { , , , ,0,0, , , }
    \setrow { , , ,0,0, , , , }
    \setrow { , , ,0,0}
    \setrow { , , , , , , , , }
    
\end{tikzpicture}}
\caption{In green we show an example of alliance, with all its cells initially in state $0$. The cells explicitly labeled with a $0$ highlight an alliance which satisfies Item 2 of the inductive definition. The green cells without a written $0$ are part of the alliance due to the union  of additional $2\times 2$ blocks to the alliance formed by the labeled cells. The green cells that do \textbf{not} contain a written $0$ do not form an alliance by themselves; however, the union of alliances is not required to be disjoint. The isolated green cells could be or not included as part of the alliance, but this makes no difference, since in this dynamics we are only interested in the connected component (of cells in state $0$) that contains the target cell.}
\label{fig:alianza}
\end{figure}



As a consequence, the prediction problem reduces to deciding whether the target cell belongs to an alliance, and we have the following lemma.

\begin{lemma}\label{lem:f1}
  For the rule $f_1$ FA-pred $\in \textbf{NL}$.
\end{lemma}

We study the case of background state $s = 1$; the other case is analogous. 

\begin{proof}
    To solve the prediction problem, we provide a nondeterministic algorithm \Cref{alg:f1} that runs in logarithmic space with respect to the size of the rectangle and checks whether the target cell belongs to an alliance.
    
    Determining whether a cell lies in a $2 \times 2$ block alliance is straightforward and is handled in Step \textbf{1}. Otherwise, the algorithm nondeterministically guesses a staircase whose endpoints are $2 \times 2$ blocks, which contains the target cell $x$, and such that all of its cells are in state $0$.
 
    More precisely, in Step \textbf{4} the algorithm selects a $2 \times 2$ block that will serve as one endpoint of the staircase. This step does not require explicitly searching for such a block; instead, it suffices to nondeterministically choose a cell and proceed only if it belongs to a$2 \times 2$ block. The values $\lambda_1,\lambda_2$ encode the direction in which the staircase of cells in state $0$ is explored. During the \textbf{while}-loop of Step \textbf{8}, the algorithm traverses the staircase by alternately taking one horizontal and one vertical step, following the directions specified by the $\lambda$ values. In Step \textbf{7}, the variables $s$ are initialized; these indicate which part of the grid is currently being inspected along the staircase. The values of $s$ are updated in Steps \textbf{12} and \textbf{13}.

    If at any point the staircase ceases to consist exclusively of cells in state $0$, the computation is immediately rejected, as specified in Steps \textbf{9} and \textbf{13}, since this would mean that the guessed staircase is invalid. The variable $f$ records whether the target cell $x$ has already been visited. If the traversal reaches a $2 \times 2$ block after the cell $x$ has been encountered, the algorithm accepts.

    The traversal of the staircase can be performed while storing only the current coordinates, which requires logarithmic space in the size of the rectangle.
\end{proof}

A deterministic algorithm could also be constructed as in \Cref{alg:f1det} to detect the presence of such a staircase, by replacing the \textbf{while}-loop with an explicit universal iteration over all cells (restarting every store value before every iteration) and by testing all possible choices of the parameters $\lambda$. However, the nondeterministic approach described above is simpler and sufficient to show that the decision problem lies below \textbf{P}, and is therefore very unlikely to be \textbf{P}-complete.
\begin{algorithm}[t]
	\SetAlgoLined
	\KwIn{Initial rectangular finite support of the configuration $c$. Cell $x$}
	\KwOut{\texttt{accept} if $x$ belong to an alliance as defined in \Cref{characaliance}}
    
	\If{$x$ is on a $2\times2$ block of value $0$ in $c$}{
	    \texttt{accept}\;
	}
	
    \For{ cell $a = (a_1,a_2)$ in a $2\times2$ block of value $0$}{
        \For{ $\lambda_1,\lambda_2\in \{-1,1\}$}{
            $f \leftarrow \texttt{False}$\;
            $s_1,s_2 \leftarrow a_1,a_2$\;
            \While{$(s_1,s_2)$ is a coordinate inside the rectangle}{
                \textbf{if} $c(s_1,s_2)=1$ \textbf{then} \texttt{break}\;
                \textbf{if} $(s_1,s_2)=x$ \textbf{then} $f \longleftarrow \texttt{True}$\;
                \textbf{if} $(s_1,s_2)$ is at $2\times2$ block of value $0$ and $f = \texttt{True} \textbf{ then } \texttt{accept} $\;
                \textbf{if} $c(s_1+\lambda_1,s_2) = 1$ \textbf{then} \texttt{break}\;
                $(s_1,s_2) \leftarrow (s_1+\lambda_1,s_2+\lambda_2)$
                }}}
    \texttt{reject}
	\caption{Deterministic algorithm to decide if $x$ belong to an alliance as defined in \Cref{characaliance}}
	\label{alg:f1det}
\end{algorithm}

\subsection{Rule \texorpdfstring{$f_2$}{TEXT}}

\label{secf2}

In this section, we discuss the prediction problem for the FA associated with $f_2$ (See \Cref{tab:rules})  with background state equal 1. For this rule, we first observe that, given a configuration~$c$, the only nontrivial \emph{alliances} (consisting of cells in state~$0$, since the dynamics is freezing and all cells in state~$1$ trivially form alliances) are isolated cells in state~$0$ whose neighbors in the von~Neumann neighborhood are all in state~$1$ (see~\Cref{fig:f2stable},\Cref{fig:Ejf2}). The dynamics converges to a configuration consisting of isolated cells in state~$0$ surrounded by cells in state~$1$ (see~\Cref{fig:f2stable}).

\begin{figure}
\centering\resizebox{0.2\textwidth}{!}{
\begin{tikzpicture}[
  every node/.append style={font=\relsize{+3}}
]

\draw[draw opacity=0,
      fill={rgb,255:red,184; green,233; blue,134},
      fill opacity=1]
      (1,3) rectangle (2,4);

\draw[draw opacity=0,
      fill={rgb,255:red,184; green,233; blue,134},
      fill opacity=1]
      (2,2) rectangle (3,3);

    \draw[line width=0.1mm] (0, 0) grid (5, 5);

    \setcounter{rows}{5}
    \setcounter{row}{0}
    \setrow {1,1,1,1,1}
    \setrow {1,0,1,1,1}
    \setrow {1,1,0,1,1}
    \setrow {1,1,1,1,1}
    \setrow {1,1,1,1,1};

\end{tikzpicture}}
\caption{Example of an alliance, and the unique type of configuration to which the
dynamics converges: a sea of $1$s with stable issolated $0$s.}
\label{fig:f2stable}
\end{figure}
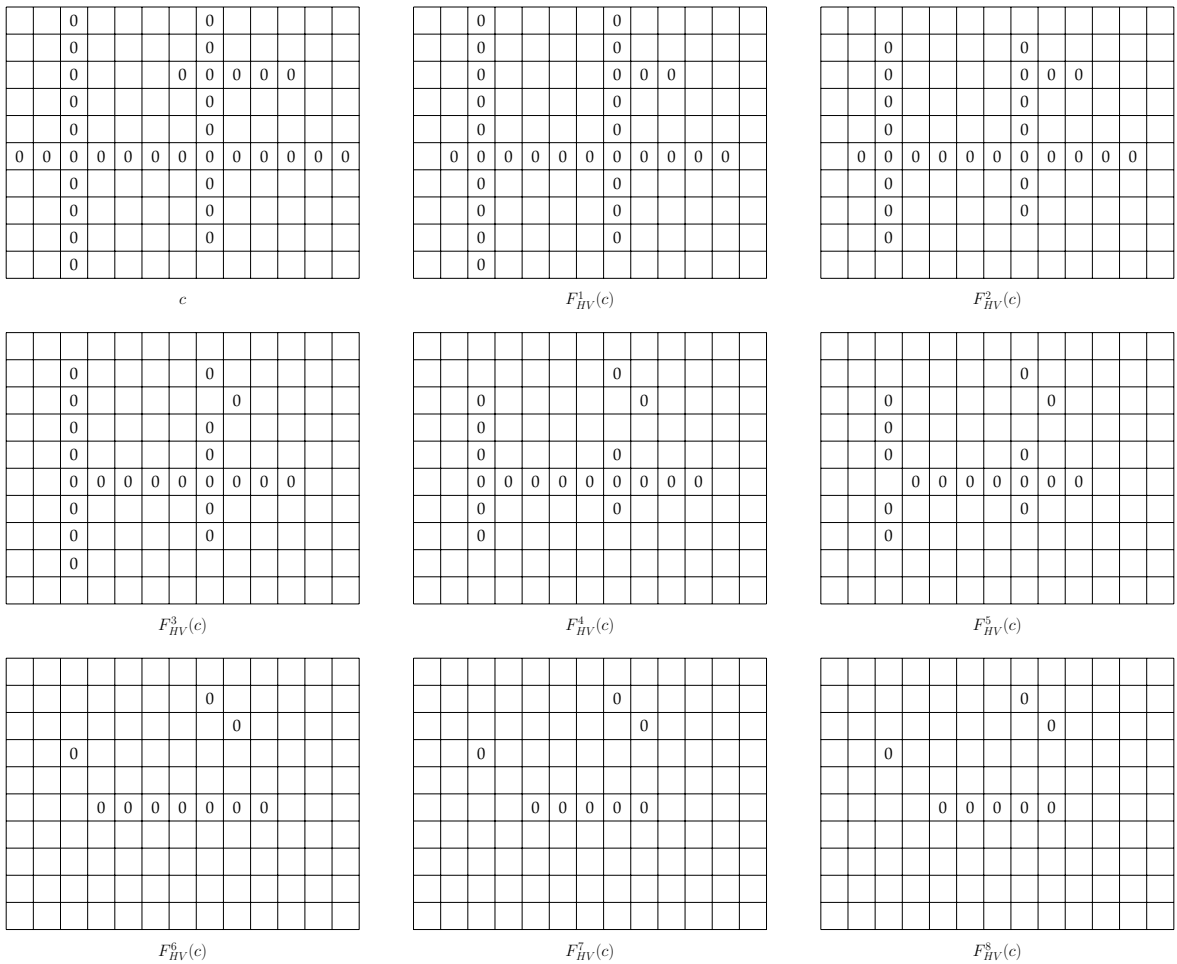

\begin{figure}
\centering\resizebox{\textwidth}{!}{
\begin{tikzpicture}[
  every node/.append style={font=\relsize{+3}}
]

\draw[line width=0.1mm] (0, 0) grid (13, 10);

\setcounter{cols}{0}
\setcounter{rows}{34}
\setcounter{row}{0}
\setrow { , ,0, , , , ,0, , , , }
\setrow { , ,0, , , , ,0, , , , }
\setrow { , ,0, , , ,0,0,0,0,0, }
\setrow { , ,0, , , , ,0, , , , }
\setrow { , ,0, , , , ,0, , , ,}
\setrow {0,0,0,0,0,0,0,0,0,0,0,0,0}
\setrow { , ,0, , , , ,0, , , }
\setrow { , ,0, , , , ,0, , , , }
\setrow { , ,0, , , , ,0, , , , }
\setrow { , ,0, , , , , , , , , }
;

\draw[line width=0.1mm] (15, 0) grid (28, 10);

\setcounter{cols}{15}
\setcounter{rows}{34}
\setcounter{row}{0}
\setrow { , ,0, , , , ,0, , , , }
\setrow { , ,0, , , , ,0, , , , }
\setrow { , ,0, , , , ,0,0,0, , }
\setrow { , ,0, , , , ,0, , , , }
\setrow { , ,0, , , , ,0, , , ,}
\setrow { ,0,0,0,0,0,0,0,0,0,0,0, }
\setrow { , ,0, , , , ,0, , , }
\setrow { , ,0, , , , ,0, , , , }
\setrow { , ,0, , , , ,0, , , , }
\setrow { , ,0, , , , , , , , , }
;

\draw[line width=0.1mm] (30, 0) grid (43, 10);

\setcounter{cols}{30}
\setcounter{rows}{34}
\setcounter{row}{0}
\setrow { , , , , , , , , , , , }
\setrow { , ,0, , , , ,0, , , , }
\setrow { , ,0, , , , ,0,0,0, , }
\setrow { , ,0, , , , ,0, , , , }
\setrow { , ,0, , , , ,0, , , ,}
\setrow { ,0,0,0,0,0,0,0,0,0,0,0, }
\setrow { , ,0, , , , ,0, , , }
\setrow { , ,0, , , , ,0, , , , }
\setrow { , ,0, , , , , , , , , }
\setrow { , , , , , , , , , , , }
;

\draw[line width=0.1mm] (0, 12) grid (13, 22);
\draw[line width=0.1mm] (15,12) grid (28, 22);
\draw[line width=0.1mm] (30,12) grid (43, 22);

\setcounter{cols}{0}
\setcounter{rows}{22}
\setcounter{row}{0}
\setrow { , , , , , , , , , , , }
\setrow { , ,0, , , , ,0, , , , }
\setrow { , ,0, , , , , ,0, , , }
\setrow { , ,0, , , , ,0, , , , }
\setrow { , ,0, , , , ,0, , , ,}
\setrow { , ,0,0,0,0,0,0,0,0,0, , }
\setrow { , ,0, , , , ,0, , , }
\setrow { , ,0, , , , ,0, , , , }
\setrow { , ,0, , , , , , , , , }
\setrow { , , , , , , , , , , , }
;

\setcounter{cols}{15}
\setcounter{rows}{22}
\setcounter{row}{0}
\setrow { , , , , , , , , , , , }
\setrow { , , , , , , ,0, , , , }
\setrow { , ,0, , , , , ,0, , , }
\setrow { , ,0, , , , , , , , , }
\setrow { , ,0, , , , ,0, , , ,}
\setrow { , ,0,0,0,0,0,0,0,0,0, , }
\setrow { , ,0, , , , ,0, , , }
\setrow { , ,0, , , , , , , , , }
\setrow { , , , , , , , , , , , }
\setrow { , , , , , , , , , , , }
;
\setcounter{cols}{30}
\setcounter{rows}{22}
\setcounter{row}{0}
\setrow { , , , , , , , , , , , }
\setrow { , , , , , , ,0, , , , }
\setrow { , ,0, , , , , ,0, , , }
\setrow { , ,0, , , , , , , , , }
\setrow { , ,0, , , , ,0, , , ,}
\setrow { , , ,0,0,0,0,0,0,0, , , }
\setrow { , ,0, , , , ,0, , , }
\setrow { , ,0, , , , , , , , , }
\setrow { , , , , , , , , , , , }
\setrow { , , , , , , , , , , , }
;

\draw[line width=0.1mm] (0, 24) grid (13, 34);
\draw[line width=0.1mm] (15,24) grid (28, 34);
\draw[line width=0.1mm] (30,24) grid (43, 34);

\setcounter{cols}{0}
\setcounter{rows}{10}
\setcounter{row}{0}
\setrow { , , , , , , , , , , , }
\setrow { , , , , , , ,0, , , , }
\setrow { , , , , , , , ,0, , , }
\setrow { , ,0, , , , , , , , , }
\setrow { , , , , , , , , , , ,}
\setrow { , , ,0,0,0,0,0,0,0, , , }
\setrow { , , , , , , , , , , }
\setrow { , , , , , , , , , , , }
\setrow { , , , , , , , , , , , }
\setrow { , , , , , , , , , , , }
;

\setcounter{cols}{15}
\setcounter{rows}{10}
\setcounter{row}{0}
\setrow { , , , , , , , , , , , }
\setrow { , , , , , , ,0, , , , }
\setrow { , , , , , , , ,0, , , }
\setrow { , ,0, , , , , , , , , }
\setrow { , , , , , , , , , , ,}
\setrow { , , , ,0,0,0,0,0, , , , }
\setrow { , , , , , , , , , , }
\setrow { , , , , , , , , , , , }
\setrow { , , , , , , , , , , , }
\setrow { , , , , , , , , , , , }
;

\setcounter{cols}{30}
\setcounter{rows}{10}
\setcounter{row}{0}
\setrow { , , , , , , , , , , , }
\setrow { , , , , , , ,0, , , , }
\setrow { , , , , , , , ,0, , , }
\setrow { , ,0, , , , , , , , , }
\setrow { , , , , , , , , , , ,}
\setrow { , , , ,0,0,0,0,0, , , , }
\setrow { , , , , , , , , , , }
\setrow { , , , , , , , , , , , }
\setrow { , , , , , , , , , , , }
\setrow { , , , , , , , , , , , }
;

\node at (6.5,-0.8) {$F^6_{HV}(c)$};
\node at (21.5,-0.8) {$F^7_{HV}(c)$};
\node at (36.5,-0.8) {$F^8_{HV}(c)$};

\node at (6.5,11.2) {$F^3_{HV}(c)$};
\node at (21.5,11.2) {$F^4_{HV}(c)$};
\node at (36.5,11.2) {$F^5_{HV}(c)$};

\node at (6.5,23.2) {$c$};
\node at (21.5,23.2) {$F^1_{HV}(c)$};
\node at (36.5,23.2) {$F^2_{HV}(c)$};
\end{tikzpicture}}

\caption{Successive updates of the fungal automaton associated with $f_2$ from a given initial configuration $c$. At each iteration, only cells with exactly one neighbor in state~$0$ and the other in state~$1$ change state.}
\label{fig:Ejf2}
\end{figure}

The configuration $c$ shown in \Cref{fig:Ejf2} is of particular interest, as it satisfies a property that simplifies its analysis: for every pair of cells in state~$0$, there exists a unique sequence of adjacent cells in state~$0$ connecting them. We refer to such configurations as \emph{acyclic configurations}. The behavior on these configurations is well characterized. As illustrated in the figure, the rows and columns are stable except at their endpoints. However, the most important result concerning the FA-pred problem for $f_2$ is provided by \Cref{teof2}, which applies to automata networks. In order to use it, we must define an equivalent easy to compute freezing automata network.

First, we observe that for every fungal automaton, one can define an associated cellular automaton with equivalent dynamics by composing the global update function. In particular, the global update rule
\[
F^2(c) = F_V(F_H(c))
\]
is induced by a local rule $f_2'$ defined over the Moore neighborhood $M$. Hence, we may consider the cellular automaton $(\{0,1\}, M, f_2')$ (see \Cref{secfunggen}).

More generally, if the word $Z$ defining the fungal automaton has length greater than two, the global rule of the associated cellular automaton would be of the form $F^{|Z|}$, where $|Z|$ denotes the length of $Z$, and the corresponding local function would have domain $M'$, a suitably larger neighborhood. However, this definition becomes somewhat artificial for fungal automata defined by arbitrary words $Z$, as it loses the structural simplicity that we aim to preserve in the model.

In the present case, we nevertheless adopt this definition, since it allows us to apply the result of \cite{paramcom}. Once the associated cellular automaton of a fungal automaton has been defined, we turn to the definition of its associated automata network.

\begin{definition}
    For a configuration $c$ whose cells with value different from $1$ are contained in a rectangle $R \subseteq \mathbb{Z}^2$, the \emph{underlying graph} is the graph $(V, E_H \cup E_V \cup E_D)$, where:\begin{align*}
        && V &= \{x \in \Z^2 \,|\, c(x) = 0\}\\
        && E_H &= \{xy \,|\, x,y \text{ are horizontal neighbors}\}\\
        && E_V &= \{xy \,|\, x,y \text{ are vertical neighbors}\}\\
        && E_D &= \{xy \,|\, \exists z \text{ such that }zx \in E_H \text{ and }zy \in E_V\}\\
    \end{align*}
\end{definition}

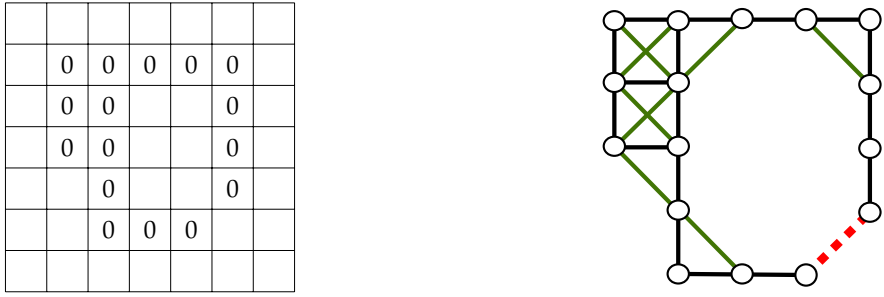
\begin{figure}[t]
\begin{minipage}[t][0.2\textheight][c]{0.5\textwidth}
\centering
\centering\resizebox{0.5\textwidth}{!}{\begin{tikzpicture}[
  every node/.append style={font=\relsize{+3}}
]
    
    \draw[line width=0.1mm] (0, 0) grid (7, 7);
    \setcounter{cols}{0}
    \setcounter{rows}{7}
    \setcounter{row}{0}
    \setrow { , , , , }
    \setrow { ,0,0,0,0,0, }
    \setrow { ,0,0, , ,0, }
    \setrow { ,0,0, , ,0, }
    \setrow { , ,0, , ,0, };
    \setrow { , ,0,0,0, , };

\end{tikzpicture}}\\
\end{minipage}
\begin{minipage}[t][0.2\textheight][c]{0.5\textwidth}
\centering
\centering\resizebox{0.5\textwidth}{!}{\tikzset{every picture/.style={line width=0.75pt}} 

\begin{tikzpicture}[x=0.75pt,y=0.75pt,yscale=-1,xscale=1]

\draw [color={rgb, 255:red, 65; green, 117; blue, 5 }  ,draw opacity=1 ][line width=1.5]    (75,144) -- (105,174.5) ;
\draw [color={rgb, 255:red, 65; green, 117; blue, 5 }  ,draw opacity=1 ][line width=1.5]    (75,84.5) -- (105,55) ;
\draw [color={rgb, 255:red, 65; green, 117; blue, 5 }  ,draw opacity=1 ][line width=1.5]    (135,55) -- (165,85.5) ;
\draw [color={rgb, 255:red, 0; green, 0; blue, 0 }  ,draw opacity=1 ][line width=1.5]    (45,84.5) -- (75,84.5) ;
\draw [color={rgb, 255:red, 0; green, 0; blue, 0 }  ,draw opacity=1 ][line width=1.5]    (45,115) -- (75,115) ;
\draw [color={rgb, 255:red, 65; green, 117; blue, 5 }  ,draw opacity=1 ][line width=1.5]    (75,55) -- (60.89,68.88) -- (45,84.5) ;
\draw [color={rgb, 255:red, 65; green, 117; blue, 5 }  ,draw opacity=1 ][line width=1.5]    (75,85.5) -- (60.89,99.38) -- (45,115) ;
\draw [color={rgb, 255:red, 65; green, 117; blue, 5 }  ,draw opacity=1 ][line width=1.5]    (45,55) -- (75,85.5) ;
\draw [color={rgb, 255:red, 65; green, 117; blue, 5 }  ,draw opacity=1 ][line width=1.5]    (45,84) -- (75,114.5) ;
\draw [color={rgb, 255:red, 65; green, 117; blue, 5 }  ,draw opacity=1 ][line width=1.5]    (45,114.5) -- (75,145) ;
\draw [color={rgb, 255:red, 0; green, 0; blue, 0 }  ,draw opacity=1 ][line width=1.5]    (45,55.5) -- (45,114.5) ;
\draw [color={rgb, 255:red, 0; green, 0; blue, 0 }  ,draw opacity=1 ][line width=1.5]    (75,55.5) -- (75,174.5) ;
\draw [color={rgb, 255:red, 0; green, 0; blue, 0 }  ,draw opacity=1 ][line width=1.5]    (165,56) -- (165,145.5) ;
\draw [color={rgb, 255:red, 0; green, 0; blue, 0 }  ,draw opacity=1 ][line width=1.5]    (75,174.5) -- (135,175) ;
\draw [color={rgb, 255:red, 0; green, 0; blue, 0 }  ,draw opacity=1 ][line width=1.5]    (45,55) -- (165,55) ;
\draw [color={rgb, 255:red, 255; green, 0; blue, 0 }  ,draw opacity=1 ][line width=3]  [dash pattern={on 3.38pt off 3.27pt}]  (135,175) -- (165,145.5) ;
\draw  [fill={rgb, 255:red, 255; green, 255; blue, 255 }  ,fill opacity=1 ] (40,55.5) .. controls (40,53.01) and (42.24,51) .. (45,51) .. controls (47.76,51) and (50,53.01) .. (50,55.5) .. controls (50,57.99) and (47.76,60) .. (45,60) .. controls (42.24,60) and (40,57.99) .. (40,55.5) -- cycle ;
\draw  [fill={rgb, 255:red, 255; green, 255; blue, 255 }  ,fill opacity=1 ] (40,84.5) .. controls (40,82.01) and (42.24,80) .. (45,80) .. controls (47.76,80) and (50,82.01) .. (50,84.5) .. controls (50,86.99) and (47.76,89) .. (45,89) .. controls (42.24,89) and (40,86.99) .. (40,84.5) -- cycle ;
\draw  [fill={rgb, 255:red, 255; green, 255; blue, 255 }  ,fill opacity=1 ] (70,84.5) .. controls (70,82.01) and (72.24,80) .. (75,80) .. controls (77.76,80) and (80,82.01) .. (80,84.5) .. controls (80,86.99) and (77.76,89) .. (75,89) .. controls (72.24,89) and (70,86.99) .. (70,84.5) -- cycle ;
\draw  [fill={rgb, 255:red, 255; green, 255; blue, 255 }  ,fill opacity=1 ] (70,55.5) .. controls (70,53.01) and (72.24,51) .. (75,51) .. controls (77.76,51) and (80,53.01) .. (80,55.5) .. controls (80,57.99) and (77.76,60) .. (75,60) .. controls (72.24,60) and (70,57.99) .. (70,55.5) -- cycle ;
\draw  [fill={rgb, 255:red, 255; green, 255; blue, 255 }  ,fill opacity=1 ] (40,114.5) .. controls (40,112.01) and (42.24,110) .. (45,110) .. controls (47.76,110) and (50,112.01) .. (50,114.5) .. controls (50,116.99) and (47.76,119) .. (45,119) .. controls (42.24,119) and (40,116.99) .. (40,114.5) -- cycle ;
\draw  [fill={rgb, 255:red, 255; green, 255; blue, 255 }  ,fill opacity=1 ] (70,114.5) .. controls (70,112.01) and (72.24,110) .. (75,110) .. controls (77.76,110) and (80,112.01) .. (80,114.5) .. controls (80,116.99) and (77.76,119) .. (75,119) .. controls (72.24,119) and (70,116.99) .. (70,114.5) -- cycle ;
\draw  [fill={rgb, 255:red, 255; green, 255; blue, 255 }  ,fill opacity=1 ] (70,144.5) .. controls (70,142.01) and (72.24,140) .. (75,140) .. controls (77.76,140) and (80,142.01) .. (80,144.5) .. controls (80,146.99) and (77.76,149) .. (75,149) .. controls (72.24,149) and (70,146.99) .. (70,144.5) -- cycle ;
\draw  [fill={rgb, 255:red, 255; green, 255; blue, 255 }  ,fill opacity=1 ] (70,174.5) .. controls (70,172.01) and (72.24,170) .. (75,170) .. controls (77.76,170) and (80,172.01) .. (80,174.5) .. controls (80,176.99) and (77.76,179) .. (75,179) .. controls (72.24,179) and (70,176.99) .. (70,174.5) -- cycle ;
\draw  [fill={rgb, 255:red, 255; green, 255; blue, 255 }  ,fill opacity=1 ] (100,174.5) .. controls (100,172.01) and (102.24,170) .. (105,170) .. controls (107.76,170) and (110,172.01) .. (110,174.5) .. controls (110,176.99) and (107.76,179) .. (105,179) .. controls (102.24,179) and (100,176.99) .. (100,174.5) -- cycle ;
\draw  [fill={rgb, 255:red, 255; green, 255; blue, 255 }  ,fill opacity=1 ] (130,175) .. controls (130,177.76) and (132.24,180) .. (135,180) .. controls (137.76,180) and (140,177.76) .. (140,175) .. controls (140,172.24) and (137.76,170) .. (135,170) .. controls (132.24,170) and (130,172.24) .. (130,175) -- cycle ;
\draw  [fill={rgb, 255:red, 255; green, 255; blue, 255 }  ,fill opacity=1 ] (100,54.5) .. controls (100,52.01) and (102.24,50) .. (105,50) .. controls (107.76,50) and (110,52.01) .. (110,54.5) .. controls (110,56.99) and (107.76,59) .. (105,59) .. controls (102.24,59) and (100,56.99) .. (100,54.5) -- cycle ;
\draw  [fill={rgb, 255:red, 255; green, 255; blue, 255 }  ,fill opacity=1 ] (130,55) .. controls (130,57.76) and (132.24,60) .. (135,60) .. controls (137.76,60) and (140,57.76) .. (140,55) .. controls (140,52.24) and (137.76,50) .. (135,50) .. controls (132.24,50) and (130,52.24) .. (130,55) -- cycle ;
\draw  [fill={rgb, 255:red, 255; green, 255; blue, 255 }  ,fill opacity=1 ] (160,55) .. controls (160,57.76) and (162.24,60) .. (165,60) .. controls (167.76,60) and (170,57.76) .. (170,55) .. controls (170,52.24) and (167.76,50) .. (165,50) .. controls (162.24,50) and (160,52.24) .. (160,55) -- cycle ;
\draw  [fill={rgb, 255:red, 255; green, 255; blue, 255 }  ,fill opacity=1 ] (160,85.5) .. controls (160,83.01) and (162.24,81) .. (165,81) .. controls (167.76,81) and (170,83.01) .. (170,85.5) .. controls (170,87.99) and (167.76,90) .. (165,90) .. controls (162.24,90) and (160,87.99) .. (160,85.5) -- cycle ;
\draw  [fill={rgb, 255:red, 255; green, 255; blue, 255 }  ,fill opacity=1 ] (160,115.5) .. controls (160,113.01) and (162.24,111) .. (165,111) .. controls (167.76,111) and (170,113.01) .. (170,115.5) .. controls (170,117.99) and (167.76,120) .. (165,120) .. controls (162.24,120) and (160,117.99) .. (160,115.5) -- cycle ;
\draw  [fill={rgb, 255:red, 255; green, 255; blue, 255 }  ,fill opacity=1 ] (160,145.5) .. controls (160,143.01) and (162.24,141) .. (165,141) .. controls (167.76,141) and (170,143.01) .. (170,145.5) .. controls (170,147.99) and (167.76,150) .. (165,150) .. controls (162.24,150) and (160,147.99) .. (160,145.5) -- cycle ;

\end{tikzpicture}}
\end{minipage}
\caption{Example of an underlying graph associated with a configuration. On the left, a configuration is shown, where unlabeled cells are assumed to have value~$1$. On the right, its underlying graph $(V,E_H \cup E_V \cup E_D)$ is shown. $E_H$ and $E_V$ edges are shown in black, and the edges in $E_D$ are shown in dark green. The definition of $E_D$ is chosen so that the dotted red edge is not included in the underlying graph of the dynamics, since the cells in state~$0$ corresponding to those two nodes do not interact in less than to updates during the dynamics.}
\label{fig:lemma3}
\end{figure}

In the definition of the underlying graph, a distinction is made between horizontal, vertical, and diagonal edges. Although this distinction is lost when the object is treated purely as a graph, it is important for defining the update function of the underlying automata network that we study.

\begin{definition} 
The \emph{underlying automata network} of the fungal automaton associated with $f_2$ on a configuration $c$ is the automata network $(G,\mathcal{F})$, where $G$ is the underlying graph and $\mathcal{F}=\{\mathcal{F}_x\}_{x\in V(G)}$ is the family of local update functions defined as follows.

For each configuration $c \in \{0,1\}^{V(G)}$ and each $x \in V(G)$, we define a configuration $c'$ on a $3\times 3$ grid (that is, on the set $\{(i,j)\mid i,j\in\{-1,0,1\}\}$) as follows:

\begin{itemize}
    \item The central value is given by $c(x)$, that is, $c'(0,0)=c(x)$.

    \item The horizontal cells are determined from the values in $c$ of the vertices connected to $x$ by edges in $E_H$. More precisely, if there exists one (resp.\ two) edges in $E_H$ connecting $x$ to one (resp.\ two) vertices $y$ (resp.\ $y,z$) with $c(y)=0$ (resp.\ $c(y)=c(z)=0$), then $c'(-1,0)=0$ (resp.\ $c'(-1,0)=c'(1,0)=0$).

    \item The vertical cells are defined analogously using the edges in $E_V$. That is, if there exists one (resp.\ two) edges in $E_V$ connecting $x$ to one (resp.\ two) vertices $y$ (resp.\ $y,z$) with $c(y)=0$ (resp.\ $c(y)=c(z)=0$), then $c'(0,-1)=0$ (resp.\ $c'(0,-1)=c'(0,1)=0$).

    \item The diagonal cells are determined from the edges in $E_D$ in a way that preserves consistency with the previous orientations. In particular, if we consider the underlying graph induced by the resulting grid configuration $c'$, it must coincide locally with the pattern of neighbors $y$ of $x$ such that $c(y)=0$.

    \item Every other cell on the grid has value 1.
\end{itemize}

We then define
\[
\mathcal{F}_x(c) = f_2'(c').
\]
In this way, the underlying automata network mimics the behavior of the cellular automaton associated with the fungal automaton, as if the dynamics were taking place on a graph rather than on the regular grid. The construction of the configuration $c'$ is illustrated in \Cref{fig:underlyngan}.
\end{definition}

Although multiple symmetric embeddings of the local neighborhood of $x$ into the $3\times 3$ grid are possible, the resulting updates coincide. This is because the rule is derived from a one-dimensional totalistic rule applied along the horizontal and vertical directions. Hence, the construction is well defined. See \Cref{fig:underlyngan}.

\begin{figure}[t]
\begin{minipage}[t][0.2\textheight][c]{0.5\textwidth}
\centering
\centering\resizebox{0.5\textwidth}{!}{\tikzset{every picture/.style={line width=0.75pt}} 

\begin{tikzpicture}[x=0.75pt,y=0.75pt,yscale=-1,xscale=1]

\draw [color={rgb, 255:red, 65; green, 117; blue, 5 }  ,draw opacity=1 ][line width=1.5]    (60,50) -- (180,170) ;
\draw [color={rgb, 255:red, 65; green, 117; blue, 5 }  ,draw opacity=1 ][line width=1.5]    (180,110) -- (120,50) ;
\draw [color={rgb, 255:red, 0; green, 0; blue, 0 }  ,draw opacity=1 ][line width=1.5]    (180,110) -- (180,170) ;
\draw [color={rgb, 255:red, 65; green, 117; blue, 5 }  ,draw opacity=1 ][line width=1.5]    (120,170) -- (180,110) ;
\draw [color={rgb, 255:red, 0; green, 0; blue, 0 }  ,draw opacity=1 ][line width=1.5]    (60,110) -- (180,110) ;
\draw [color={rgb, 255:red, 0; green, 0; blue, 0 }  ,draw opacity=1 ][line width=1.5]    (120,170) -- (180,170) ;
\draw [color={rgb, 255:red, 65; green, 117; blue, 5 }  ,draw opacity=1 ][line width=1.5]    (120,50) -- (60,110) ;
\draw [color={rgb, 255:red, 65; green, 117; blue, 5 }  ,draw opacity=1 ][line width=1.5]    (60,110) -- (120,170) ;
\draw [color={rgb, 255:red, 0; green, 0; blue, 0 }  ,draw opacity=1 ][line width=1.5]    (60,50) -- (60,110) ;
\draw [color={rgb, 255:red, 0; green, 0; blue, 0 }  ,draw opacity=1 ][line width=1.5]    (120,50.87) -- (120,170) ;
\draw [color={rgb, 255:red, 0; green, 0; blue, 0 }  ,draw opacity=1 ][line width=1.5]    (60,50) -- (120,50) ;
\draw  [fill={rgb, 255:red, 255; green, 255; blue, 255 }  ,fill opacity=1 ] (50,50) .. controls (50,44.48) and (54.48,40) .. (60,40) .. controls (65.52,40) and (70,44.48) .. (70,50) .. controls (70,55.52) and (65.52,60) .. (60,60) .. controls (54.48,60) and (50,55.52) .. (50,50) -- cycle ;
\draw  [fill={rgb, 255:red, 255; green, 255; blue, 255 }  ,fill opacity=1 ] (50,110) .. controls (50,104.48) and (54.48,100) .. (60,100) .. controls (65.52,100) and (70,104.48) .. (70,110) .. controls (70,115.52) and (65.52,120) .. (60,120) .. controls (54.48,120) and (50,115.52) .. (50,110) -- cycle ;
\draw  [fill={rgb, 255:red, 255; green, 255; blue, 255 }  ,fill opacity=1 ] (110,50) .. controls (110,44.48) and (114.48,40) .. (120,40) .. controls (125.52,40) and (130,44.48) .. (130,50) .. controls (130,55.52) and (125.52,60) .. (120,60) .. controls (114.48,60) and (110,55.52) .. (110,50) -- cycle ;
\draw  [fill={rgb, 255:red, 255; green, 255; blue, 255 }  ,fill opacity=1 ] (110,110) .. controls (110,104.48) and (114.48,100) .. (120,100) .. controls (125.52,100) and (130,104.48) .. (130,110) .. controls (130,115.52) and (125.52,120) .. (120,120) .. controls (114.48,120) and (110,115.52) .. (110,110) -- cycle ;
\draw  [fill={rgb, 255:red, 255; green, 255; blue, 255 }  ,fill opacity=1 ] (110,170) .. controls (110,164.48) and (114.48,160) .. (120,160) .. controls (125.52,160) and (130,164.48) .. (130,170) .. controls (130,175.52) and (125.52,180) .. (120,180) .. controls (114.48,180) and (110,175.52) .. (110,170) -- cycle ;
\draw  [fill={rgb, 255:red, 255; green, 255; blue, 255 }  ,fill opacity=1 ] (170,170) .. controls (170,164.48) and (174.48,160) .. (180,160) .. controls (185.52,160) and (190,164.48) .. (190,170) .. controls (190,175.52) and (185.52,180) .. (180,180) .. controls (174.48,180) and (170,175.52) .. (170,170) -- cycle ;
\draw  [fill={rgb, 255:red, 255; green, 255; blue, 255 }  ,fill opacity=1 ] (170,110) .. controls (170,104.48) and (174.48,100) .. (180,100) .. controls (185.52,100) and (190,104.48) .. (190,110) .. controls (190,115.52) and (185.52,120) .. (180,120) .. controls (174.48,120) and (170,115.52) .. (170,110) -- cycle ;

\draw (115,105.4) node [anchor=north west][inner sep=0.75pt] {$x$};
\draw (115,163.4) node [anchor=north west][inner sep=0.75pt] {$1$};
\draw (175,104.4) node [anchor=north west][inner sep=0.75pt] {$0$};
\draw (175,164.4) node [anchor=north west][inner sep=0.75pt] {$0$};
\draw (55,104.4) node [anchor=north west][inner sep=0.75pt] {$0$};
\draw (55,44.4) node [anchor=north west][inner sep=0.75pt] {$0$};
\draw (115,44.4) node [anchor=north west][inner sep=0.75pt] {$0$};
\end{tikzpicture}}\\
\end{minipage}
\begin{minipage}[t][0.2\textheight][c]{0.5\textwidth}
\centering
\centering\resizebox{0.5\textwidth}{!}{\begin{tikzpicture}[
  every node/.append style={font=\relsize{+3}}
]
    
    \draw[line width=0.1mm] (0, 0) grid (3, 3);
    \setcounter{cols}{0}
    \setcounter{rows}{3}
    \setcounter{row}{0}
    \setrow {0,0,1}
    \setrow {0,$c(x)$,0}
    \setrow {1,1,0}

    \draw[line width=0.1mm] (4, 0) grid (7, 3);
    \setcounter{cols}{4}
    \setcounter{rows}{3}
    \setcounter{row}{0}
    \setrow {1,0,0}
    \setrow {0,$c(x)$,0}
    \setrow {0,1,1}

    \draw[line width=0.1mm] (8, 0) grid (11, 3);
    \setcounter{cols}{8}
    \setcounter{rows}{3}
    \setcounter{row}{0}
    \setrow {0,0,0}
    \setrow {0,$c(x)$,0}
    \setrow {1,1,1}

\end{tikzpicture}}
\end{minipage}
\caption{Example of the grid constructed for a node $x$ in the definition of the local update function. The graph (left) shows the node $x$ together with its neighbors and their values in the configuration $c$, with edges oriented according to whether they are horizontal or vertical. The three $3\times 3$ grids (right) illustrate possible configurations $g$ induced by this local structure. The first two grids satisfy the definition and yield the same value under $f_2'$, showing that the construction is consistent. The last grid is not valid, as its underlying graph would induce a local structure different form that of the original graph.}
\label{fig:underlyngan}
\end{figure}
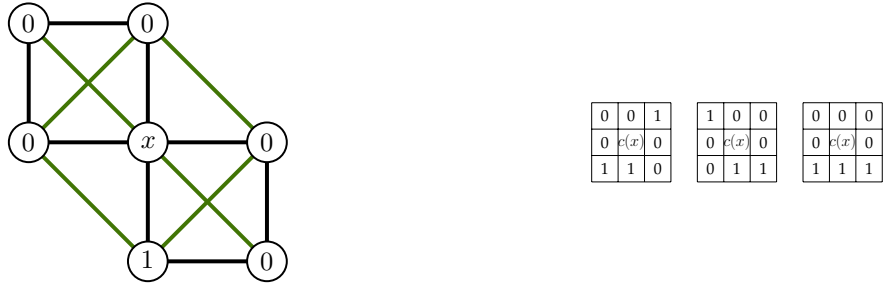

\begin{lemma}
    For underlying graph of bounded treewidth, the prediction problem \Cref{ANpred} for the underlying automata network lies in \textbf{NC}.
\end{lemma}

\begin{proof}
    This is due a directed application of \Cref{teof2} since the underlying automata network has bounded degree and is freezing.
\end{proof}

\section{Fungal freezing majority rule of radius 1.5}
\label{s:maj-radius1.5}

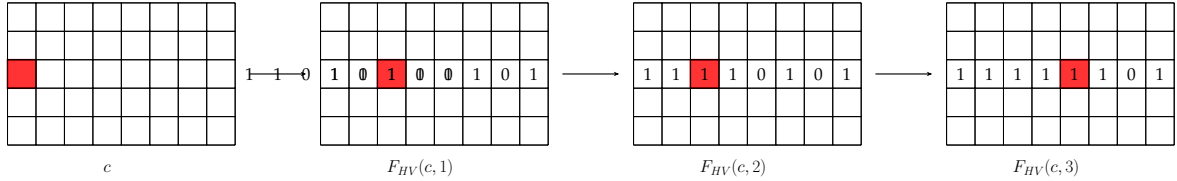
\begin{figure}[t]
\centering
\resizebox{\textwidth}{!}{%
\begin{tikzpicture}[
  every node/.append style={font=\relsize{+3}}
]
\draw[draw=red, fill=red, fill opacity=0.8, line width=0.6mm]
  (0,12) rectangle (1,13);

\draw[draw=red, fill=red, fill opacity=0.8, line width=0.6mm]
  (13,12) rectangle (14,13);

\draw[draw=red, fill=red, fill opacity=0.8, line width=0.6mm]
  (24,12) rectangle (25,13);

\draw[draw=red, fill=red, fill opacity=0.8, line width=0.6mm]
  (37,12) rectangle (38,13);
  
    \draw[line width=0.1mm] (0, 10) grid (8, 15);

    \setcounter{rows}{15}
    \setcounter{row}{0}
    \setrow { , , , , }
    \setrow { , , , , }
    \setrow {1,1,0,1,0,1,0,1}
    \setrow { , , , , };
    
    \draw[line width=0.1mm] (11, 10) grid (19, 15);

    \setcounter{cols}{11}
    \setcounter{rows}{15}
    \setcounter{row}{0}
    \setrow { , , , , }
    \setrow { , , , , }
    \setrow {1,1,1,1,0,1,0,1}
    \setrow { , , , , };
    
    \draw[line width=0.1mm] (22, 10) grid (30, 15);

    \setcounter{cols}{22}
    \setcounter{rows}{15}
    \setcounter{row}{0}
    \setrow { , , , , }
    \setrow { , , , , }
    \setrow {1,1,1,1,0,1,0,1}
    \setrow { , , , , };

    \draw[line width=0.1mm] (33, 10) grid (41, 15);

    \setcounter{cols}{33}
    \setcounter{rows}{15}
    \setcounter{row}{0}
    \setrow { , , , , }
    \setrow { , , , , }
    \setrow {1,1,1,1,1,1,0,1}
    \setrow { , , , , };

\draw[line width=0.1mm] (0, 10) grid (8, 15);
\draw[line width=0.1mm] (11, 10) grid (19, 15);
\draw[line width=0.1mm] (22, 10) grid (30, 15);
\draw[line width=0.1mm] (33, 10) grid (41, 15);

\draw[-{Stealth[length=2mm]}, line width=0.4mm]
  (8.5,12.5) -- (10.5,12.5);

\draw[-{Stealth[length=2mm]}, line width=0.4mm]
  (19.5,12.5) -- (21.5,12.5);

\draw[-{Stealth[length=2mm]}, line width=0.4mm]
  (30.5,12.5) -- (32.5,12.5);


\node at (3.5,9.2) {$c$};
\node at (14.5,9.2) {$F_{HV}(c,1)$};
\node at (25.5,9.2) {$F_{HV}(c,2)$};
\node at (36.5,9.2) {$F_{HV}(c,3)$};
\setcounter{cols}{0}

\end{tikzpicture}
}
\caption{Wire gadget represented in configuration $c$, where unlabeled cells are assumed to have value $0$. In red, we highlight the cell containing an additional \(1\) with respect to the original wire configuration, in order to denote the \emph{propagation} of the \emph{signal} along the wire.}
\label{fig:wireresad}
\end{figure}

\begin{figure}[t]
\centering
\resizebox{0.3\textwidth}{!}{\begin{tikzpicture}[
  every node/.append style={font=\relsize{+3}}
]
\draw[line width=0.1mm] (0, 10) grid (8, 16);

    \setcounter{rows}{16}
    \setcounter{row}{0}
    \setrow { , , , , }
    \setrow {0,1,0,1,0,1,0,1};
    \setrow { , , , , };
    \setrow { , , , , };
    \setrow {1,0,1,0,1,0,1,0}

\end{tikzpicture}}
\caption{Two wires of different parity.
Note that all cells with value $1$ both wires share the same parity of the sum of their coordinates (either even or odd), while they are in rows of different parity}
\label{fig:paritireal}
\end{figure}

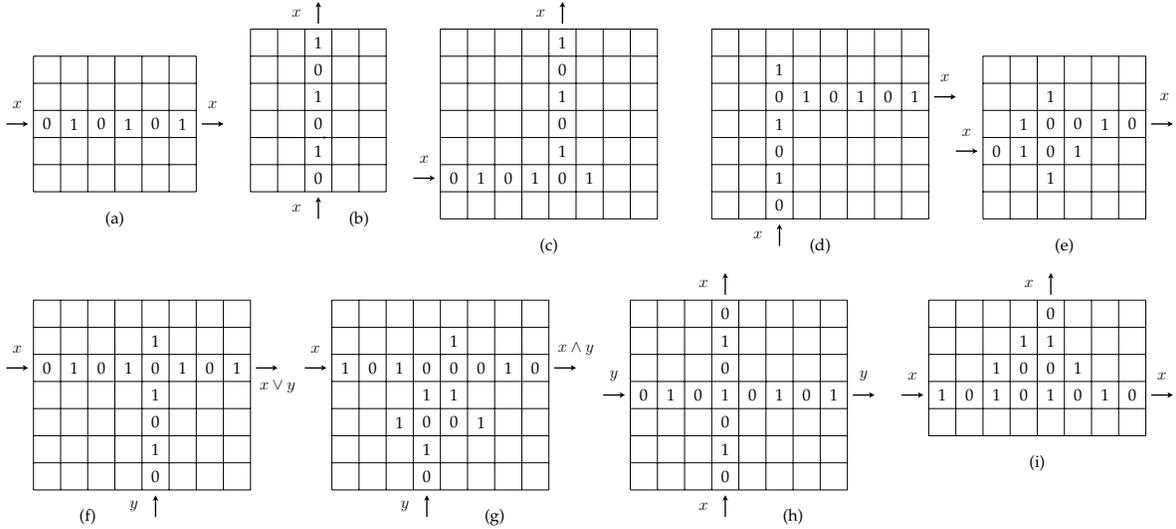
\begin{figure}[t]
\centering
\resizebox{\textwidth}{!}{%
\begin{tikzpicture}[
  every node/.append style={font=\relsize{+3}}
]

    \draw[-{Stealth[width=2mm,length=2mm]}, line width=0.5mm] (-1,12.5) -- (-0.2,12.5);
    \draw[-{Stealth[width=2mm,length=2mm]}, line width=0.5mm] (6.2,12.5) -- (7,12.5);
    
    \draw[-{Stealth[width=2mm,length=2mm]}, line width=0.5mm] (10.5,9) -- (10.5,9.8);
    \draw[-{Stealth[width=2mm,length=2mm]}, line width=0.5mm] (10.5,16.2) -- (10.5,17);

    \draw[-{Stealth[width=2mm,length=2mm]}, line width=0.5mm] (14,10.5) -- (14.8,10.5);
    \draw[-{Stealth[width=2mm,length=2mm]}, line width=0.5mm] (19.5,16.2) -- (19.5,17);

    \draw[-{Stealth[width=2mm,length=2mm]}, line width=0.5mm] (27.5,8) -- (27.5,8.8);
    \draw[-{Stealth[width=2mm,length=2mm]}, line width=0.5mm] (33.2,13.5) -- (34,13.5);

    \draw[-{Stealth[width=2mm,length=2mm]}, line width=0.5mm] (34,11.5) -- (34.8,11.5);
    \draw[-{Stealth[width=2mm,length=2mm]}, line width=0.5mm] (41.2,12.5) -- (42,12.5);


    \draw[-{Stealth[width=2mm,length=2mm]}, line width=0.5mm] (-1,3.5) -- (-0.2,3.5);
    \draw[-{Stealth[width=2mm,length=2mm]}, line width=0.5mm] (10,3.5) -- (10.8,3.5);
    \draw[-{Stealth[width=2mm,length=2mm]}, line width=0.5mm] (21,2.5) -- (21.8,2.5);

    \draw[-{Stealth[width=2mm,length=2mm]}, line width=0.5mm] (8.2,3.5) -- (9,3.5);
    \draw[-{Stealth[width=2mm,length=2mm]}, line width=0.5mm] (19.2,3.5) -- (20,3.5);
    \draw[-{Stealth[width=2mm,length=2mm]}, line width=0.5mm] (30.2,2.5) -- (31,2.5);

    \draw[-{Stealth[width=2mm,length=2mm]}, line width=0.5mm] (41.2,2.5) -- (42,2.5);
    \draw[-{Stealth[width=2mm,length=2mm]}, line width=0.5mm] (32,2.5) -- (32.8,2.5);
    \node at (32.4,3.2) {$x$};
    \node at (41.6,3.2) {$x$};
    
    \draw[-{Stealth[width=2mm,length=2mm]}, line width=0.5mm] (4.5,-2) -- (4.5,-1.2);
    \draw[-{Stealth[width=2mm,length=2mm]}, line width=0.5mm] (14.5,-2) -- (14.5,-1.2);
    \draw[-{Stealth[width=2mm,length=2mm]}, line width=0.5mm] (25.5,-2) -- (25.5,-1.2);

    \draw[-{Stealth[width=2mm,length=2mm]}, line width=0.5mm] (25.5,6.2) -- (25.5,7);
    \draw[-{Stealth[width=2mm,length=2mm]}, line width=0.5mm] (37.5,6.2) -- (37.5,7);

\node at (-0.6,13.2) {$x$};
\node at (6.6,13.2) {$x$};

\node at (9.7,9.4) {$x$};
\node at (9.7,16.6) {$x$};

\node at (14.4,11.2) {$x$};
\node at (18.7,16.6) {$x$};

\node at (26.7,8.4) {$x$};
\node at (33.6,14.2) {$x$};

\node at (34.4,12.2) {$x$};
\node at (41.7,13.5) {$x$};


\node at (-0.6,4.2) {$x$};
\node at (10.4,4.2) {$x$};
\node at (21.4,3.2) {$y$};

\node at (9,2.8) {$x\lor y$};
\node at (20,4.2) {$x\land y$};
\node at (30.6,3.2) {$y$};

\node at (3.7,-1.6) {$y$};
\node at (13.7,-1.6) {$y$};
\node at (24.7,-1.6) {$x$};

\node at (24.7,6.6) {$x$};
\node at (36.7,6.6) {$x$};

    \draw[line width=0.1mm] (0, 10) grid (6, 15);

    \setcounter{rows}{15}
    \setcounter{row}{0}
    \setrow { , , , , }
    \setrow { , , , , }
    \setrow {0,1,0,1,0,1, }
    \setrow { , , , , };
    
    \draw[- stealth] (8.2,12) -- (10.8,12);

    \setcounter{cols}{7}
    \setcounter{rows}{16}
    \setcounter{row}{0}
    \setrow { , , ,1, }
    \setrow { , , ,0, }
    \setrow { , , ,1, }
    \setrow { , , ,0, }
    \setrow { , , ,1, }
    \setrow { , , ,0, }
;

    \draw[line width=0.1mm] (8, 10) grid (13, 16);

    \draw[line width=0.1mm] (15, 9) grid (23, 16);

    \setcounter{cols}{15}
    \setcounter{rows}{16}
    \setcounter{row}{0}
    \setrow { , , , ,1}
    \setrow { , , , ,0}
    \setrow { , , , ,1}
    \setrow { , , , ,0}
    \setrow { , , , ,1}
    \setrow {0,1,0,1,0,1}
    \setrow { , , , , , };

    \draw[line width=0.1mm] (25, 9) grid (33, 16);

    \setcounter{cols}{25}
    \setcounter{rows}{16}
    \setcounter{row}{0}
    \setrow { , , , , }
    \setrow { , ,1, , }
    \setrow { , ,0,1,0,1,0,1}
    \setrow { , ,1}
    \setrow { , ,0}
    \setrow { , ,1}
    \setrow { , ,0};

    \draw[line width=0.1mm] (35, 9) grid (41, 15);

    \setcounter{cols}{35}
    \setcounter{rows}{16}
    \setcounter{row}{0}
    \setrow { , , , , }
    \setrow { , }
    \setrow { , ,1,}
    \setrow { ,1,0,0,1,0}
    \setrow {0,1,0,1}
    \setrow { , ,1};

    \draw[line width=0.1mm] (0, -1) grid (8, 6);

    \setcounter{cols}{33}
    \setcounter{rows}{6}
    \setcounter{row}{0}
    \setrow { , , , ,0, }    
    \setrow { , , ,1,1, }
    \setrow { , ,1,0,0,1}
    \setrow {1,0,1,0,1,0,1,0};

    \draw[line width=0.1mm] (11, -1) grid (19, 6);

    \setcounter{cols}{22}
    \setcounter{rows}{6}
    \setcounter{row}{0}
    \setrow { , , ,0, }
    \setrow { , , ,1}
    \setrow {  , , ,0}
    \setrow {0,1,0,1,0,1,0,1}
    \setrow {  , , ,0}
    \setrow {  , , ,1}
    \setrow {  , , ,0};

    \draw[line width=0.1mm] (22, -1) grid (30, 6);

    \setcounter{cols}{11}
    \setcounter{rows}{6}
    \setcounter{row}{0}
    \setrow { , , , }
    \setrow { , , , ,1}
    \setrow {1,0,1,0,0,0,1,0}
    \setrow { , , ,1,1, , , }
    \setrow { , ,1,0,0,1, , }
    \setrow { , , ,1, , , , }
    \setrow { , , ,0};

    \draw[line width=0.1mm] (33, 1) grid (41, 6);

    \setcounter{cols}{0}
    \setcounter{rows}{6}
    \setcounter{row}{0}
    \setrow { , , , , }
    \setrow { , , , ,1}
    \setrow {0,1,0,1,0,1,0,1}
    \setrow { , , , ,1}
    \setrow { , , , ,0}
    \setrow { , , , ,1}
    \setrow { , , , ,0};

\setcounter{cols}{0}

\node at (3,9) {(a)};
\node at (12,9) {(b)};

\node at (19,8) {(c)};
\node at (29,8) {(d)};
\node at (38,8) {(e)};

\node at (2,-2) {(f)};
\node at (17,-2) {(g)};
\node at (28,-2) {(h)};
\node at (37,0) {(i)};

\end{tikzpicture}
}
\caption{In the top row, we define the following gadgets: (a) the \emph{horizontal wire}, which moves bits from left to right (see \Cref{fig:wireresad}), (b) the \emph{vertical wire}, which moves bits from bottom to top, (c,d) the \emph{turns}, which connect wires from different directions, and (e) the \emph{coordinator wire}, which changes the parity of a wire (See \Cref{fig:paritireal}). In the bottom row, from left to right, we define the gadgets for  \emph{AND} (f), which computes the $\land$ of two bits,  \emph{OR} (g), which computes the $\lor$ of two bits, the \emph{crossing} (h), which allows bits to cross, and the \emph{duplicator} (i), which duplicates bits. In the figure, each gadget is shown together with its input bits, represented by $x$ or $y$.}
\label{fig:gadgetdresad}
\end{figure}

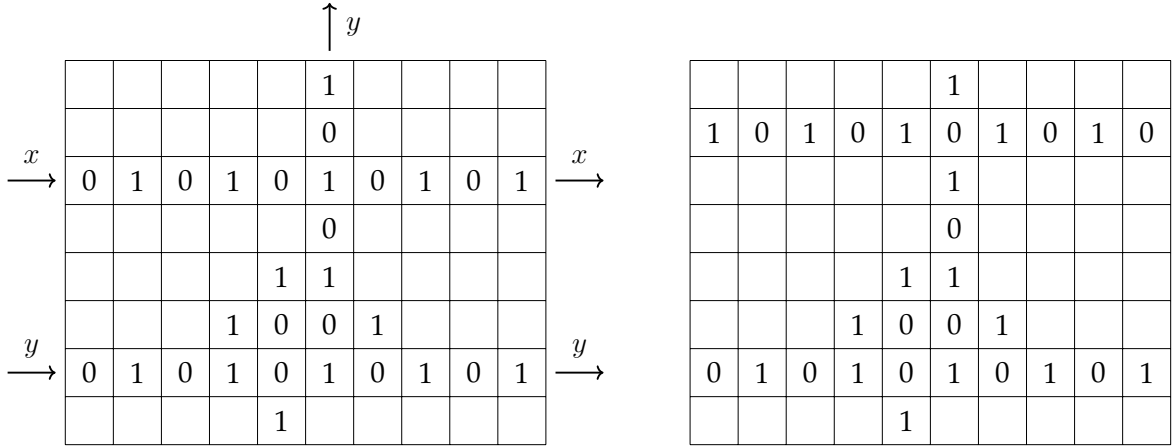
\begin{figure}[t]
\centering
\resizebox{\textwidth}{!}{%
\begin{tikzpicture}[
  every node/.append style={font=\relsize{+3}}
]
\draw[->, line width=0.4mm] (-1.2,5.5) -- (-0.2,5.5);
\node at (-0.7,6.0) {$x$};

\draw[->, line width=0.4mm] (-1.2,1.5) -- (-0.2,1.5);
\node at (-0.7,2.0) {$y$};

\draw[->, line width=0.4mm] (10.2,5.5) -- (11.2,5.5);
\node at (10.7,6.0) {$x$};

\draw[->, line width=0.4mm] (5.5,8.2) -- (5.5,9.2);
\node at (6.0,8.7) {$y$};

\draw[->, line width=0.4mm] (10.2,1.5) -- (11.2,1.5);
\node at (10.7,2.0) {$y $};





    \draw[line width=0.1mm] (0, 0) grid (10, 8);

    \setcounter{rows}{8}
    \setcounter{row}{0}

    \setrow { , , , , ,1, , };
    \setrow { , , , , ,0, , };
    \setrow {0,1,0,1,0,1,0,1,0,1};
    \setrow { , , , , ,0, , ,};
    \setrow { , , , ,1,1, , ,};
    \setrow { , , ,1,0,0,1, ,};
    \setrow {0,1,0,1,0,1,0,1,0,1};
    \setrow { , , , ,1, ,, ,};
    \setrow { , , , , , ,, , };

    \draw[line width=0.1mm] (13, 0) grid (23, 8);

    \setcounter{cols}{13}
    \setcounter{rows}{8}
    \setcounter{row}{0}
    \setrow { , , , , ,1, };
    \setrow {1,0,1,0,1,0,1,0,1,0};
    \setrow { , , , , ,1, , ,};
    \setrow { , , , , ,0, , ,};
    \setrow { , , , ,1,1, , ,};
    \setrow { , , ,1,0,0,1, ,};
    \setrow {0,1,0,1,0,1,0,1,0,1};
    \setrow { , , , ,1, ,, ,};
    \setrow { , , , , , ,, ,};
    \setrow { , , , , , ,, , };

\end{tikzpicture}
}
\caption{Crossing of wires. The figure shows two attempts of wire crossings. On the left, wires with even distance are crossed without issues with the gadget previously defined, since the propagation of a signal $x$ or $y$ along either wire will not interfere with the correct behavior of the other. On the right, we illustrate that wires with odd distance do not admit a clean crossing.}
\label{fig:pariti}
\end{figure}

\begin{figure}[t]
\centering
\resizebox{\textwidth}{!}{%
\tikzset{every picture/.style={line width=0.75pt}} 

\begin{tikzpicture}[x=0.75pt,y=0.75pt,yscale=-1,xscale=1]

\draw [color={rgb, 255:red, 74; green, 144; blue, 226 }  ,draw opacity=1 ]   (160,250) -- (177,250) ;
\draw [shift={(180,250)}, rotate = 180] [fill={rgb, 255:red, 74; green, 144; blue, 226 }  ,fill opacity=1 ][line width=0.08]  [draw opacity=0] (5.36,-2.57) -- (0,0) -- (5.36,2.57) -- cycle    ;
\draw  [color={rgb, 255:red, 0; green, 0; blue, 0 }  ,draw opacity=1 ] (280,160) -- (320,160) -- (320,210) -- (280,210) -- cycle ;
\draw [color={rgb, 255:red, 208; green, 2; blue, 27 }  ,draw opacity=1 ]   (190,177) -- (190,260) ;
\draw [color={rgb, 255:red, 74; green, 144; blue, 226 }  ,draw opacity=1 ]   (200,187) -- (200,220) ;
\draw [color={rgb, 255:red, 74; green, 144; blue, 226 }  ,draw opacity=1 ]   (200,187) -- (277,187) ;
\draw [shift={(280,187)}, rotate = 180] [fill={rgb, 255:red, 74; green, 144; blue, 226 }  ,fill opacity=1 ][line width=0.08]  [draw opacity=0] (5.36,-2.57) -- (0,0) -- (5.36,2.57) -- cycle    ;
\draw [color={rgb, 255:red, 208; green, 2; blue, 27 }  ,draw opacity=1 ]   (190,177) -- (277,177) ;
\draw [shift={(280,177)}, rotate = 180] [fill={rgb, 255:red, 208; green, 2; blue, 27 }  ,fill opacity=1 ][line width=0.08]  [draw opacity=0] (5.36,-2.57) -- (0,0) -- (5.36,2.57) -- cycle    ;
\draw [color={rgb, 255:red, 208; green, 2; blue, 27 }  ,draw opacity=1 ]   (160,260) -- (177,260) ;
\draw [shift={(180,260)}, rotate = 180] [fill={rgb, 255:red, 208; green, 2; blue, 27 }  ,fill opacity=1 ][line width=0.08]  [draw opacity=0] (5.36,-2.57) -- (0,0) -- (5.36,2.57) -- cycle    ;
\draw [color={rgb, 255:red, 208; green, 2; blue, 27 }  ,draw opacity=1 ]   (160,300) -- (177,300) ;
\draw [shift={(180,300)}, rotate = 180] [fill={rgb, 255:red, 208; green, 2; blue, 27 }  ,fill opacity=1 ][line width=0.08]  [draw opacity=0] (5.36,-2.57) -- (0,0) -- (5.36,2.57) -- cycle    ;
\draw [color={rgb, 255:red, 74; green, 144; blue, 226 }  ,draw opacity=1 ]   (180,220) -- (330,220) ;
\draw [color={rgb, 255:red, 74; green, 144; blue, 226 }  ,draw opacity=1 ]   (180,250) -- (350,250) ;
\draw  [color={rgb, 255:red, 74; green, 144; blue, 226 }  ,draw opacity=1 ][fill={rgb, 255:red, 255; green, 255; blue, 255 }  ,fill opacity=1 ] (195,220) .. controls (195,217.24) and (197.24,215) .. (200,215) .. controls (202.76,215) and (205,217.24) .. (205,220) .. controls (205,222.76) and (202.76,225) .. (200,225) .. controls (197.24,225) and (195,222.76) .. (195,220) -- cycle ;
\draw [color={rgb, 255:red, 208; green, 2; blue, 27 }  ,draw opacity=1 ]   (180,260) -- (360,260) ;
\draw [color={rgb, 255:red, 208; green, 2; blue, 27 }  ,draw opacity=1 ]   (180,300) -- (380,300) ;
\draw [color={rgb, 255:red, 74; green, 144; blue, 226 }  ,draw opacity=1 ]   (330,220) -- (330,30) ;
\draw [color={rgb, 255:red, 74; green, 144; blue, 226 }  ,draw opacity=1 ]   (350,250) -- (350,60) ;
\draw [color={rgb, 255:red, 208; green, 2; blue, 27 }  ,draw opacity=1 ]   (360,260) -- (360,70) ;
\draw [color={rgb, 255:red, 208; green, 2; blue, 27 }  ,draw opacity=1 ]   (300,160) -- (300,100) ;
\draw [color={rgb, 255:red, 208; green, 2; blue, 27 }  ,draw opacity=1 ]   (380,300) -- (380,90) ;
\draw [color={rgb, 255:red, 74; green, 144; blue, 226 }  ,draw opacity=1 ]   (160,220) -- (177,220) ;
\draw [shift={(180,220)}, rotate = 180] [fill={rgb, 255:red, 74; green, 144; blue, 226 }  ,fill opacity=1 ][line width=0.08]  [draw opacity=0] (5.36,-2.57) -- (0,0) -- (5.36,2.57) -- cycle    ;
\draw [color={rgb, 255:red, 74; green, 144; blue, 226 }  ,draw opacity=1 ]   (330,30) -- (427,30) ;
\draw [shift={(430,30)}, rotate = 180] [fill={rgb, 255:red, 74; green, 144; blue, 226 }  ,fill opacity=1 ][line width=0.08]  [draw opacity=0] (5.36,-2.57) -- (0,0) -- (5.36,2.57) -- cycle    ;
\draw [color={rgb, 255:red, 74; green, 144; blue, 226 }  ,draw opacity=1 ]   (350,60) -- (427,60) ;
\draw [shift={(430,60)}, rotate = 180] [fill={rgb, 255:red, 74; green, 144; blue, 226 }  ,fill opacity=1 ][line width=0.08]  [draw opacity=0] (5.36,-2.57) -- (0,0) -- (5.36,2.57) -- cycle    ;
\draw [color={rgb, 255:red, 208; green, 2; blue, 27 }  ,draw opacity=1 ]   (360,70) -- (427,70) ;
\draw [shift={(430,70)}, rotate = 180] [fill={rgb, 255:red, 208; green, 2; blue, 27 }  ,fill opacity=1 ][line width=0.08]  [draw opacity=0] (5.36,-2.57) -- (0,0) -- (5.36,2.57) -- cycle    ;
\draw [color={rgb, 255:red, 208; green, 2; blue, 27 }  ,draw opacity=1 ]   (300,100) -- (427,100) ;
\draw [shift={(430,100)}, rotate = 180] [fill={rgb, 255:red, 208; green, 2; blue, 27 }  ,fill opacity=1 ][line width=0.08]  [draw opacity=0] (5.36,-2.57) -- (0,0) -- (5.36,2.57) -- cycle    ;
\draw [color={rgb, 255:red, 208; green, 2; blue, 27 }  ,draw opacity=1 ]   (380,90) -- (427,90) ;
\draw [shift={(430,90)}, rotate = 180] [fill={rgb, 255:red, 208; green, 2; blue, 27 }  ,fill opacity=1 ][line width=0.08]  [draw opacity=0] (5.36,-2.57) -- (0,0) -- (5.36,2.57) -- cycle    ;
\draw  [color={rgb, 255:red, 0; green, 0; blue, 0 }  ,draw opacity=1 ] (180,20) -- (420,20) -- (420,330) -- (180,330) -- cycle ;
\draw  [color={rgb, 255:red, 208; green, 2; blue, 27 }  ,draw opacity=1 ][fill={rgb, 255:red, 255; green, 255; blue, 255 }  ,fill opacity=1 ] (185,260) .. controls (185,257.24) and (187.24,255) .. (190,255) .. controls (192.76,255) and (195,257.24) .. (195,260) .. controls (195,262.76) and (192.76,265) .. (190,265) .. controls (187.24,265) and (185,262.76) .. (185,260) -- cycle ;
\draw [color={rgb, 255:red, 208; green, 2; blue, 27 }  ,draw opacity=1 ] [dash pattern={on 0.84pt off 2.51pt}]  (170,270) -- (170,290) ;
\draw [color={rgb, 255:red, 74; green, 144; blue, 226 }  ,draw opacity=1 ] [dash pattern={on 0.84pt off 2.51pt}]  (170,230) -- (170,240) ;
\draw [color={rgb, 255:red, 0; green, 0; blue, 0 }  ,draw opacity=1 ] [dash pattern={on 0.84pt off 2.51pt}]  (40,210) -- (280,210) ;
\draw [color={rgb, 255:red, 0; green, 0; blue, 0 }  ,draw opacity=1 ] [dash pattern={on 0.84pt off 2.51pt}]  (420,20) -- (420,0) ;
\draw [color={rgb, 255:red, 0; green, 0; blue, 0 }  ,draw opacity=1 ] [dash pattern={on 0.84pt off 2.51pt}]  (610,160) -- (320,160) ;
\draw [color={rgb, 255:red, 0; green, 0; blue, 0 }  ,draw opacity=1 ] [dash pattern={on 0.84pt off 2.51pt}]  (180,430) -- (180,330) ;

\draw (290,172.4) node [anchor=north west][inner sep=0.75pt]    {$G_{i}$};
\draw (141,209.4) node [anchor=north west][inner sep=0.75pt]    {$x_{1}$};
\draw (141,239.4) node [anchor=north west][inner sep=0.75pt]    {$x_{n}$};
\draw (230,82.4) node [anchor=north west][inner sep=0.75pt]    {$T_{i}$};
\draw (141,252.4) node [anchor=north west][inner sep=0.75pt]    {$g_{1}$};
\draw (131,289.4) node [anchor=north west][inner sep=0.75pt]    {$g_{i-1}$};
\draw (441,19.4) node [anchor=north west][inner sep=0.75pt]    {$x_{1}$};
\draw (441,49.4) node [anchor=north west][inner sep=0.75pt]    {$x_{n}$};
\draw (441,62.4) node [anchor=north west][inner sep=0.75pt]    {$g_{1}$};
\draw (440,79.4) node [anchor=north west][inner sep=0.75pt]    {$g_{i-1}$};
\draw (441,92.4) node [anchor=north west][inner sep=0.75pt]    {$g_{i}$};
\draw (531,82.4) node [anchor=north west][inner sep=0.75pt]    {$T_{i+1}$};
\draw (72,332.4) node [anchor=north west][inner sep=0.75pt]    {$T_{i-1}$};

\end{tikzpicture}
}
\caption{For each $i \in [m]$, each $T_i$ represents a tile in this embedding of a circuit. Each tile has a fixed size $O(m)$, which provides enough space for all wires carrying the bits corresponding to each input (depicted as blue lines) and to the outputs of all preceding gates (depicted in red). Tile $T_i$ computes the output of gate $G_i$. The figure illustrates an example in which gate $g_i$ receives the wires corresponding to the input node $x_1$ and to the gate $g_1$, meaning that, in the description of the circuit, we had the tuple $(g_i,t,x_1,g_1)$ with $t$ some logical gate. The tiles must be arranged with strictly increasing height, since wires cannot propagate signals downward. As the size of the tiles is fixed and the positions of the wires inside each tile are fixed as well, the vertical offset between consecutive tiles is constant.}
\label{fig:tilesresad}
\end{figure}
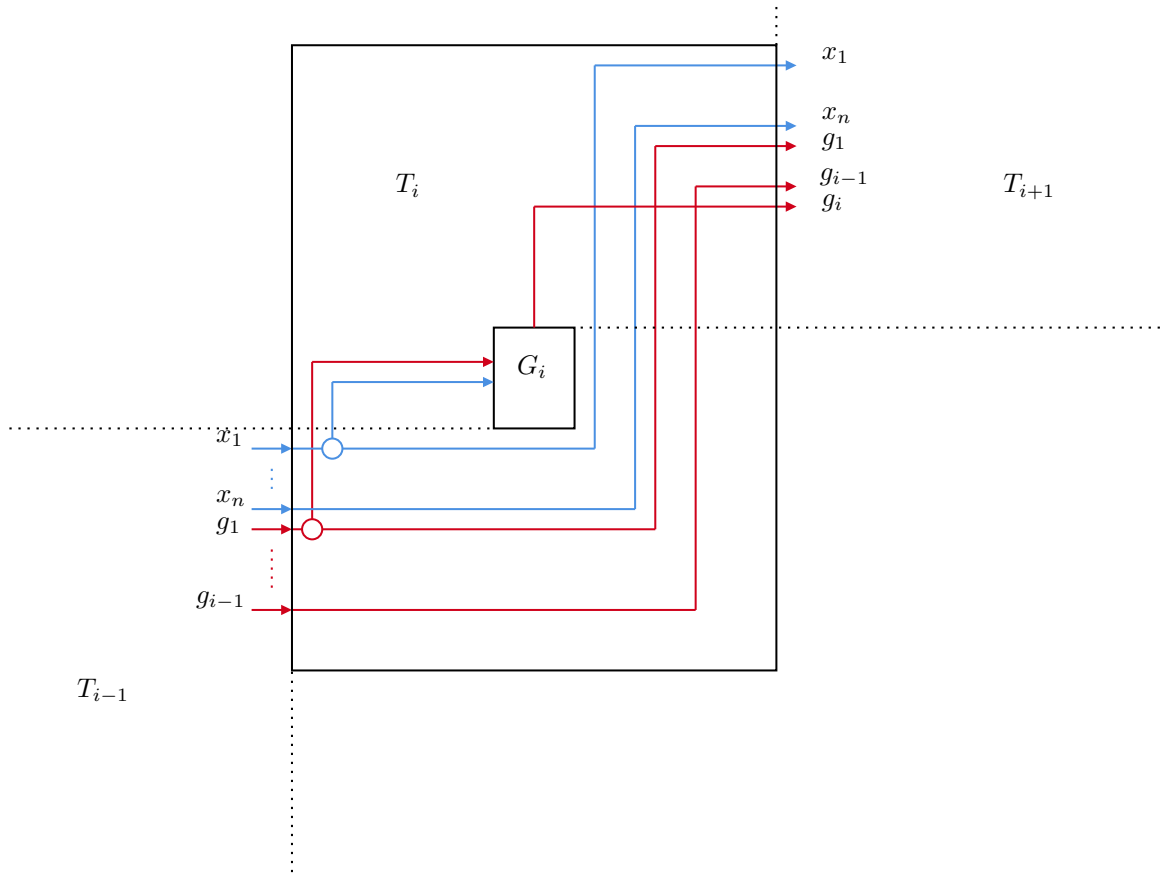

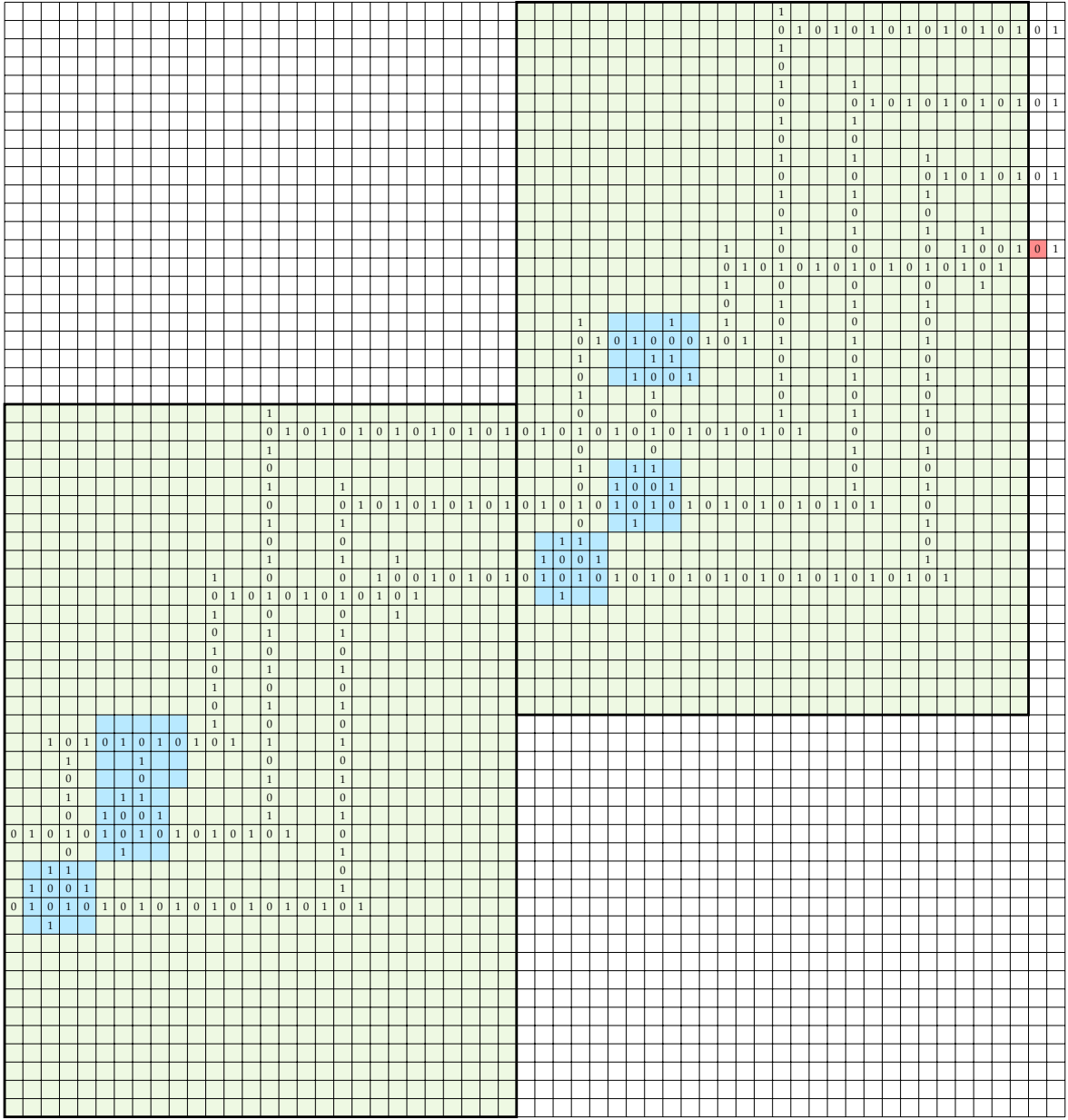
\begin{figure}[t]
\centering
\resizebox{\textwidth}{!}{%
\begin{tikzpicture}[
  every node/.append style={font=\relsize{+3}}
]

\draw[draw opacity=0,
      fill={rgb,255:red,184; green,233; blue,134},
      fill opacity=0.23]
      (0,0) rectangle (28,39);
      
\draw[draw opacity=1,
      fill={rgb,255:red,0; green,0; blue,0},
      fill opacity=0][line width=4]
      (0,0) rectangle (28,39);

\draw[draw opacity=0,
      fill={rgb,255:red,184; green,233; blue,134},
      fill opacity=0.23]
      (28,22) rectangle (56,61);

\draw[draw opacity=0,
      fill={rgb,255:red,184; green,233; blue,255},
      fill opacity=1]
      (1,10) rectangle (5,14);
      
\draw[draw opacity=0,
      fill={rgb,255:red,184; green,233; blue,255},
      fill opacity=1]
      (5,14) rectangle (9,18);

\draw[draw opacity=0,
      fill={rgb,255:red,184; green,233; blue,255},
      fill opacity=1]
      (5,18) rectangle (10,22);

\draw[draw opacity=0,
      fill={rgb,255:red,184; green,233; blue,255},
      fill opacity=1]
      (29,28) rectangle (33,32);
      
\draw[draw opacity=0,
      fill={rgb,255:red,184; green,233; blue,255},
      fill opacity=1]
      (33,32) rectangle (37,36);
\draw[draw opacity=0,
      fill={rgb,255:red,184; green,233; blue,255},
      fill opacity=1]
      (33,40) rectangle (38,44);

\draw[draw opacity=0,
      fill={rgb,255:red,255; green,140; blue,140},
      fill opacity=1]
      (56,47) rectangle (57,48);

\draw[draw opacity=1,
      fill={rgb,255:red,0; green,0; blue,0},
      fill opacity=0][line width=4]
      (28,22) rectangle (56,61);


    \draw[line width=0.01mm] (0, 0) grid (58, 61);
    
    \setcounter{cols}{0}
    \setcounter{rows}{61}
    \setcounter{row}{0}
    \setrow { , , , , , , , , , , , , , , , , , , , , , , , , , , , , , , , , , , , , , , , , , ,1, , , , , , , , , , , , };
    \setrow { , , , , , , , , , , , , , , , , , , , , , , , , , , , , , , , , , , , , , , , , , ,0,1,0,1,0,1,0,1,0,1,0,1,0,1,0,1};
    \setrow { , , , , , , , , , , , , , , , , , , , , , , , , , , , , , , , , , , , , , , , , , ,1, , , , , , , , , , , , };
    \setrow { , , , , , , , , , , , , , , , , , , , , , , , , , , , , , , , , , , , , , , , , , ,0, , , , , , , , , , , , };
    \setrow { , , , , , , , , , , , , , , , , , , , , , , , , , , , , , , , , , , , , , , , , , ,1, , , ,1, , , , , , , , };
    \setrow { , , , , , , , , , , , , , , , , , , , , , , , , , , , , , , , , , , , , , , , , , ,0, , , ,0,1,0,1,0,1,0,1,0,1,0,1};
    \setrow { , , , , , , , , , , , , , , , , , , , , , , , , , , , , , , , , , , , , , , , , , ,1, , , ,1, , , , , , , , };
    \setrow { , , , , , , , , , , , , , , , , , , , , , , , , , , , , , , , , , , , , , , , , , ,0, , , ,0, , , , , , , , };
    \setrow { , , , , , , , , , , , , , , , , , , , , , , , , , , , , , , , , , , , , , , , , , ,1, , , ,1, , , ,1, , , , };
    \setrow { , , , , , , , , , , , , , , , , , , , , , , , , , , , , , , , , , , , , , , , , , ,0, , , ,0, , , ,0,1,0,1,0,1,0,1};
    \setrow { , , , , , , , , , , , , , , , , , , , , , , , , , , , , , , , , , , , , , , , , , ,1, , , ,1, , , ,1, , , , };
    \setrow { , , , , , , , , , , , , , , , , , , , , , , , , , , , , , , , , , , , , , , , , , ,0, , , ,0, , , ,0, , , , };
    \setrow { , , , , , , , , , , , , , , , , , , , , , , , , , , , , , , , , , , , , , , , , , ,1, , , ,1, , , ,1, , ,1, };
    \setrow { , , , , , , , , , , , , , , , , , , , , , , , , , , , , , , , , , , , , , , ,1, , ,0, , , ,0, , , ,0, ,1,0,0,1,0,1};
    \setrow { , , , , , , , , , , , , , , , , , , , , , , , , , , , , , , , , , , , , , , ,0,1,0,1,0,1,0,1,0,1,0,1,0,1,0,1};
    \setrow { , , , , , , , , , , , , , , , , , , , , , , , , , , , , , , , , , , , , , , ,1, , ,0, , , ,0, , , ,0, , ,1};
    \setrow { , , , , , , , , , , , , , , , , , , , , , , , , , , , , , , , , , , , , , , ,0, , ,1, , , ,1, , , ,1, , };
    \setrow { , , , , , , , , , , , , , , , , , , , , , , , , , , , , , , ,1, , , , ,1, , ,1, , ,0, , , ,0, , , ,0, , };
    \setrow { , , , , , , , , , , , , , , , , , , , , , , , , , , , , , , ,0,1,0,1,0,0,0,1,0,1, ,1, , , ,1, , , ,1, , };
    \setrow { , , , , , , , , , , , , , , , , , , , , , , , , , , , , , , ,1, , , ,1,1, , , , , ,0, , , ,0, , , ,0, , };
    \setrow { , , , , , , , , , , , , , , , , , , , , , , , , , , , , , , ,0, , ,1,0,0,1, , , , ,1, , , ,1, , , ,1, , };
    \setrow { , , , , , , , , , , , , , , , , , , , , , , , , , , , , , , ,1, , , ,1, , , , , , ,0, , , ,0, , , ,0, , };
    \setrow { , , , , , , , , , , , , , ,1, , , , , , , , , , , , , , , , ,0, , , ,0, , , , , , ,1, , , ,1, , , ,1, , };
    \setrow { , , , , , , , , , , , , , ,0,1,0,1,0,1,0,1,0,1,0,1,0,1,0,1,0,1,0,1,0,1,0,1,0,1,0,1,0,1, , ,0, , , ,0, , , };
    \setrow { , , , , , , , , , , , , , ,1, , , , , , , , , , , , , , , , ,0, , , ,0, , , , , , , , , , ,1, , , ,1, , };
    \setrow { , , , , , , , , , , , , , ,0, , , , , , , , , , , , , , , , ,1, , ,1,1, , , , , , , , , , ,0, , , ,0, , };
    \setrow { , , , , , , , , , , , , , ,1, , , ,1, , , , , , , , , , , , ,0, ,1,0,0,1, , , , , , , , , ,1, , , ,1, , };
    \setrow { , , , , , , , , , , , , , ,0, , , ,0,1,0,1,0,1,0,1,0,1,0,1,0,1,0,1,0,1,0,1,0,1,0,1,0,1,0,1,0,1, , ,0, };
    \setrow { , , , , , , , , , , , , , ,1, , , ,1, , , , , , , , , , , , ,0, , ,1, , , , , , , , , , , , , , , ,1, , , };
    \setrow { , , , , , , , , , , , , , ,0, , , ,0, , , , , , , , , , , ,1,1, , , , , , , , , , , , , , , , , , ,0, , , };
    \setrow { , , , , , , , , , , , , , ,1, , , ,1, , ,1, , , , , , , ,1,0,0,1, , , , , , , , , , , , , , , , , ,1, , , };
    \setrow { , , , , , , , , , , ,1, , ,0, , , ,0, ,1,0,0,1,0,1,0,1,0,1,0,1,0,1,0,1,0,1,0,1,0,1,0,1,0,1,0,1,0,1,0,1};
    \setrow { , , , , , , , , , , ,0,1,0,1,0,1,0,1,0,1,0,1, , , , , , , ,1, , , , , , , , , , , ,};
    \setrow { , , , , , , , , , , ,1, , ,0, , , ,0, , ,1, , , , , , , , , , , , , , , , , , , , , ,};
    \setrow { , , , , , , , , , , ,0, , ,1, , , ,1, , , , , , , , , , , , , , , , , , , , , , , , ,};
    \setrow { , , , , , , , , , , ,1, , ,0, , , ,0, , , , , , , , , , , , , , , , , , , , , , , , ,};
    \setrow { , , , , , , , , , , ,0, , ,1, , , ,1, , , , , , , , , , , , , , , , , , , , , , , , ,};
    \setrow { , , , , , , , , , , ,1, , ,0, , , ,0, , , , , , , , , , , , , , , , , , , , , , , , ,};
    \setrow { , , , , , , , , , , ,0, , ,1, , , ,1, , , , , , , , , , , , , , , , , , , , , , , ,};
    \setrow { , , , , , , , , , , ,1, , ,0, , , ,0, , , , , , , , , , , , , , , , , , , , , , , , ,};
    \setrow { , ,1,0,1,0,1,0,1,0,1,0,1, ,1, , , ,1, , , , , , , , , , , , , , , , , , , , , , , ,};
    \setrow { , , ,1, , , ,1, , , , , , ,0, , , ,0, , , , , , , , , , , , , , , , , , , , , , , , ,};
    \setrow { , , ,0, , , ,0, , , , , , ,1, , , ,1, , , , , , , , , , , , , , , , , , , , , , , , ,};
    \setrow { , , ,1, , ,1,1, , , , , , ,0, , , ,0, , , , , , , , , , , , , , , , , , , , , , , , ,};
    \setrow { , , ,0, ,1,0,0,1, , , , , ,1, , , ,1, , , , , , , , , , , , , , , , , , , , , , , , , ,};
    \setrow {0,1,0,1,0,1,0,1,0,1,0,1,0,1,0,1, , ,0, , , , , , , , , , , , , , , , , , , , ,};
    \setrow { , , ,0, , ,1, , , , , , , , , , , ,1, , , , , };
    \setrow { , ,1,1, , , , , , , , , , , , , , ,0, , , , , , , , , , ,};
    \setrow { ,1,0,0,1, , , , , , , , , , , , , ,1, , , , , , , , , , ,};
    \setrow {0,1,0,1,0,1,0,1,0,1,0,1,0,1,0,1,0,1,0,1};
    \setrow { , ,1, , , , , ,};

\end{tikzpicture} 
}
\caption{Example of an embedding of the circuit with two input nodes $x,y$ computing $(x \lor y)\land y$. The first gate has incoming edges from $x$ and $y$, while the second gate receives as input the output of the first gate and the node $y$. Since the circuit contains two gates, the construction consists of two tiles connected horizontally, with the second tile placed above the first one; both tiles have the same size. In the lower-right part of each tile, space is reserved for four parallel wires. The gadgets responsible for computing the logical operations and signal duplicators are highlighted in blue. The cells in state $0$ located outside the second gadget represent the output wires of each node. The one marked in red corresponds to the output of the output node, and is the target cell in the prediction problem.}
\label{fig:tilesresadejemplo}
\end{figure}

In the previous section, we studied the fungal automata induced by the one-dimensional majority rule ($f_1$) and showed that its prediction problem (\textsc{FA-pred}) lies in \textbf{NL}. In this section, we continue the study of FA derived from the \emph{freezing majority} rule. 

Our main motivation is to understand how the complexity of the prediction problem changes when such a rule is composed in two dimensions via the fungal construction. In particular, we show that when the underlying one-dimensional freezing majority rule has radius $1.5$, the associated fungal automaton has a \textbf{P}-complete prediction problem. This stands in sharp contrast with the one-dimensional case, where the prediction problem remains in \textbf{NL}.

In this section, we study a one-dimensional freezing totalistic automata with radius greater than one for which the fungal automata prediction problem becomes \textbf{P}-complete. We use the term \emph{becomes} because it is known \cite{goles2021freezing} that the prediction problem for all one-dimensional freezing CA lies in \textbf{NL}.

In particular, given the automata $\mathcal{M} = (\{0,1\},\{-2,-1,0,1\},f)$, (the definition of CA is recalled in \Cref{sec:ca}) where $f$ is the one-dimensional mayority freezing rule, that is:
\begin{equation*}
f(x_1,x_2,x_3,x_4) =
\begin{cases}
    1 & \text{ if } x_3 = 1, \\
    1 & \text{ if } x_1+x_2+x_4 > 2,\\
    0 & \text{ otherwise}.
\end{cases}
\end{equation*}

We study the FA associated to $\mathcal{M}, $$\mathcal{F} = (f, \{0,1\}, M, HV, f_H, f_V)$ (See \Cref{fig:explregla}) for which we have:
\begin{theorem}\label{teo2}
    For the rule $f$ defined above, the FA-pred is $\textbf{P}-$Complete for background state 0.
\end{theorem}

\begin{figure}[t]
\centering
\resizebox{0.2\textwidth}{!}{%
\begin{tikzpicture}[
  every node/.append style={font=\relsize{+3}}
]

\draw[draw opacity=0,
      fill={rgb,255:red,184; green,233; blue,134},
      fill opacity=0.43]
      (0,2) rectangle (4,3);

\draw[draw opacity=0,
      fill={rgb,255:red,184; green,233; blue,134},
      fill opacity=0.43]
      (2,0) rectangle (3,4);
      
    \draw[line width=0.1mm] (0, 0) grid (5, 5);
    
    \setcounter{cols}{0}
    
    \setcounter{rows}{5}
    \setcounter{row}{0}
    \setrow { , , , , }
    \setrow { , , , , }
    \setrow { , ,$x$, , }
    \setrow { , , , , }
    \setrow { , , , , };


\end{tikzpicture}
}
\caption{Neighborhood $M$ for the FA $\mathcal{F}$. For a cell $x$ at coordinate $(0,0)$, the neighborhood $M$ is highlighted in green.}
\label{fig:vecindadmay15}
\end{figure}
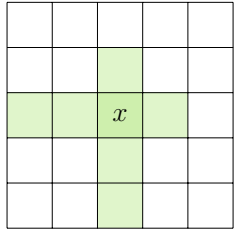

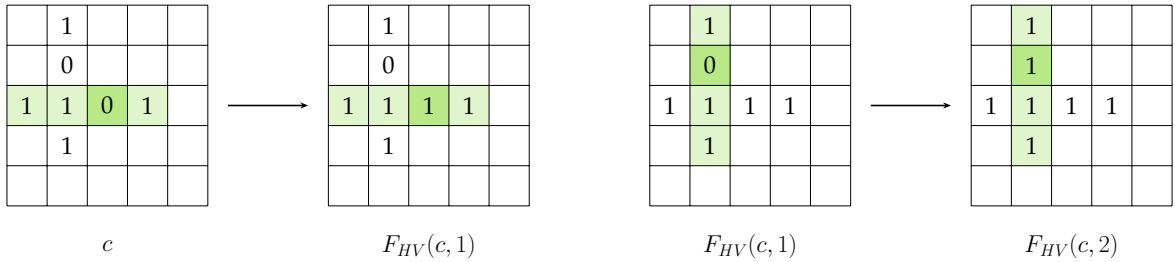
\begin{figure}[t]
\centering
\resizebox{\textwidth}{!}{%
\begin{tikzpicture}[
  every node/.append style={font=\relsize{+3}}
]

\draw[draw opacity=0,
      fill={rgb,255:red,184; green,233; blue,134},
      fill opacity=0.43]
      (0,2) rectangle (4,3);

\draw[draw opacity=0,
      fill={rgb,255:red,184; green,233; blue,134},
      fill opacity=1]
      (2,2) rectangle (3,3);

\draw[draw opacity=0,
      fill={rgb,255:red,184; green,233; blue,134},
      fill opacity=0.43]
      (8,2) rectangle (12,3);

\draw[draw opacity=0,
      fill={rgb,255:red,184; green,233; blue,134},
      fill opacity=1]
      (10,2) rectangle (11,3);
\draw[draw opacity=0,
      fill={rgb,255:red,184; green,233; blue,134},
      fill opacity=0.43]
      (17,1) rectangle (18,5);

\draw[draw opacity=0,
      fill={rgb,255:red,184; green,233; blue,134},
      fill opacity=1]
      (17,3) rectangle (18,4);

\draw[draw opacity=0,
      fill={rgb,255:red,184; green,233; blue,134},
      fill opacity=0.43]
      (25,1) rectangle (26,5);

\draw[draw opacity=0,
      fill={rgb,255:red,184; green,233; blue,134},
      fill opacity=1]
      (25,3) rectangle (26,4);

\draw[-{Stealth[length=2mm]}, line width=0.4mm]
      (5.5,2.5) -- (7.5,2.5);

\draw[-{Stealth[length=2mm]}, line width=0.4mm]
      (21.5,2.5) -- (23.5,2.5);

    \draw[line width=0.1mm] (0, 0) grid (5, 5);

    \setcounter{rows}{5}
    \setcounter{row}{0}
    \setcounter{cols}{0}
    \setrow { ,1, , , }
    \setrow { ,0, , , }
    \setrow {1,1,0,1, }
    \setrow { ,1, , , }
    \setrow { , , , , };
    
    \draw[line width=0.1mm] (8, 0) grid (13, 5);

    \setcounter{cols}{8}
    \setcounter{rows}{5}
    \setcounter{row}{0}
    \setcounter{rows}{5}
    \setcounter{row}{0}
    \setrow { ,1, , , }
    \setrow { ,0, , , }
    \setrow {1,1,1,1, }
    \setrow { ,1, , , }
    \setrow { , , , , };
    
    \draw[line width=0.1mm] (16, 0) grid (21, 5);

    \setcounter{cols}{16}
    \setcounter{rows}{5}
    \setcounter{row}{0}
    \setcounter{rows}{5}
    \setcounter{row}{0}
    \setrow { ,1, , , }
    \setrow { ,0, , , }
    \setrow {1,1,1,1, }
    \setrow { ,1, , , }
    \setrow { , , , , };

    \draw[line width=0.1mm] (24, 0) grid (29, 5);

    \setcounter{cols}{24}
    \setcounter{rows}{5}
    \setcounter{row}{0}
    \setcounter{rows}{5}
    \setcounter{row}{0}
    \setrow { ,1, , , }
    \setrow { ,1, , , }
    \setrow {1,1,1,1, }
    \setrow { ,1, , , }
    \setrow { , , , , };

\setcounter{cols}{0}

\node at (2.5,-1) {$c$};
\node at (10.5,-1) {$F_{HV}(c,1)$};
\node at (18.5,-1) {$F_{HV}(c,1)$};
\node at (26.5,-1) {$F_{HV}(c,2)$};

\end{tikzpicture}
}
\caption{Updates of the FA $\mathcal{F}$. On the left, the initial configuration $c$ is shown, where unlabeled cells are assumed to have value $0$. We note that $F_{HV}(c,1) = f_H(c)$ and $F_{HV}(c,2) = f_V\big(F_{HV}(c,1)\big)$. This means that the first update is performed with respect to the horizontal neighborhood, while the second update is performed with respect to the vertical neighborhood. The remaining grids represent updates of $c$ according to the global function $F_{HV}$. In the center of the grids from the left, the cell highlighted in dark green changed according to the local rule $f$ using its horizontal neighborhood (cells highlighted in light green). On the right, the dark green cell is updated according to $f$ with respect to its vertical neighbors (highlighted in light green). }
\label{fig:explregla}
\end{figure}

We prove the case were the background state has value 0, as the case with background state at value 1 uses the same construction but with a surrounding of cells with value 0 for the rectangle. That the prediction problem lies in \textbf{P} is immediate, since freezing
dynamics converge in a linear number of steps.
\noindent
\begin{proof}
The reduction is from \textsc{MCVP}. Given a circuit $C$ with $m$ gates and $n$ input gates, described by tuples $(g, t, g_1, g_2)$ as in \Cref{sec:CVP} ($g$ is the index of the gate, $t$ specifies the type of the gate, and $g_1$ and $g_2$ are the indices of the gates providing the inputs to gate $g$), we embed the Boolean circuit using a collection of gadgets, including \emph{wires} (See \Cref{fig:wireresad}), \emph{turns}, \emph{AND} gates, \emph{OR} gates, \emph{coordinator wires}  (used to change the parity of a wire, see \Cref{fig:paritireal}), \emph{duplicators}, and \emph{crossings}. All these gadgets are depicted in \Cref{fig:gadgetdresad} and are embedded in the grid. When a wire has at least one of its cells originally in state $0$ taking value $1$, we say that the wire carries a bit $1$; otherwise, it carries a bit $0$.

The circuit here is not delay-sensitive: all gadgets function correctly regardless of the arrival times of the signals. With this observation in mind, we follow and overall embedding strategy with Tiles arranged sequentially, where each tile computes the output of a logical gate. Some aspects to take in consideration with wires when we embeed the circuit on a grid with this gadgets is:
\begin{itemize}

\item Let a grid configuration using the structures defined in \Cref{fig:gadgetdresad} be given. We say that the wires are \emph{correctly embedded} if all cells with initial value $1$ along each wire have the same parity in the sum of their coordinates. More precisely, suppose a wire connects the points $(0,0)$ and $(0,10)$, and the cell at $(0,0)$ has value $0$. By the definition of the wires, the cell at $(0,1)$ has value $1$, the cell at $(0,2)$ has value $0$, the cell at $(0,3)$ has value $1$ and so on. In this wire, every cell initially in state $1$ satisfies that the sum of its coordinates is odd (e.g., $0+1=1$, $0+3=3$, etc.). This parity condition must be satisfied by all wires in order for the configuration to be correctly embedded.

\item In a correctly embedded circuit, parallel wires must be placed at an even distance from each other in order to allow crossings after a duplicator, a key feature required for the general tile-based circuit construction described below. The reason for this parity constraint is illustrated in \Cref{fig:paritireal}. To shift the height of a wire by one unit, we use the \emph{coordinator wire} described in \Cref{fig:gadgetdresad}.

\end{itemize}

The embedding of the circuit is carried out using \emph{tiles}, where each tile computes the value of a gate $g_i$, and the tiles are arranged sequentially, one after another. Wires transmit signals only horizontally to the right and vertically upwards. Consequently, the placement of tiles follows the layout shown in \Cref{fig:tilesresad}. Each tile has fixed size $O(m)$, where $m$ is the number of nodes in the circuit. The reduction is computable in logarithmic space and proceeds tile by tile. The target cell for the prediction problem is chosen to be the output cell of the final gadget.
\end{proof}


\textbf{Credits:} This work was supported by ANID, Fondecyt Regular 1250984. Also by Project HORIZON-MSCA-2022-SE-01 project 101131549 ``Application-driven Challenges for Automata Networks and Complex Systems (ACANCOS)’', Project ANR ``Ordinal Time Computations'' (OTC) ANR-24-CE48-0335-01 and ALARICE ANR-24-CE48-7504.

\end{document}